\pdfoutput=1
\documentclass[11pt]{amsart}

\usepackage[margin=3cm]{geometry}
\usepackage[verbose=true]{microtype}
\usepackage{amsmath}
\usepackage{amssymb}
\usepackage{amsthm}
\usepackage{bbm}
\usepackage{mathtools}
\usepackage[foot]{amsaddr}
\usepackage{color}
\usepackage[colorlinks,linkcolor=black,citecolor=black, bookmarks=false]{hyperref}
\usepackage[all]{xy}
\usepackage{subcaption}
\usepackage{array}
\usepackage{enumerate}
\usepackage{enumitem}
\usepackage{tikz}

\newtheorem{theorem}{Theorem}
\newtheorem{lemma}[theorem]{Lemma}
\newtheorem{corollary}[theorem]{Corollary}
\newtheorem{proposition}[theorem]{Proposition}
\newtheorem{problem}[theorem]{Problem}
\newtheorem{conjecture}[theorem]{Conjecture}
\newtheorem{definition}[theorem]{Definition}
\newtheorem{question}[theorem]{Question}

\theoremstyle{definition}

\newtheorem{example}[theorem]{Example}
\newtheorem{remark}[theorem]{Remark}

\newcommand{\be}{\begin{eqnarray}}
\newcommand{\ee}{\end{eqnarray}}
\newcommand{\ben}{\begin{enumerate}}
\newcommand{\een}{\end{enumerate}}

\newcommand{\ba}{\begin{array}}
\newcommand{\ea}{\end{array}}

\newcommand{\tr}{\mathrm{tr}}

\newcommand{\nn}{\nonumber}

\newcommand{\psdrank}{\textrm{psd-rank}}
\newcommand{\rank}{\textrm{rank}}
\newcommand{\seprank}{\textrm{sep-rank}}
\newcommand{\sosrank}{\textrm{sos-rank}}
\newcommand{\nnrank}{\textrm{nn-rank}}

\makeatletter 
\newcommand\nobreakpar{\par\nobreak\@afterheading} 
\makeatother

\newcommand\xqed[1]{%
  \leavevmode\unskip\penalty9999 \hbox{}\nobreak\hfill
  \quad\hbox{#1}}
\newcommand\demo{\xqed{$\triangle$}}

\newcommand{\deco}{}

\usepackage[textsize=footnotesize]{todonotes}

\begin{document}

\title[Polynomial decompositions with invariance and positivity]{%
Polynomial decompositions with invariance and positivity inspired by tensors}

\author{Gemma De las Cuevas}
\address{Institute for Theoretical Physics, Technikerstr.\ 21a,  A-6020 Innsbruck, Austria}

\author{Andreas Klingler}
\address{Institute for Theoretical Physics, Technikerstr.\ 21a,  A-6020 Innsbruck, Austria}
\email{andreas.klingler@uibk.ac.at}

\author{Tim Netzer}
\address{Department of Mathematics, Technikerstr.\ 13,  A-6020 Innsbruck, Austria}
\date{\today}



\begin{abstract} 
We present a framework to decompose real multivariate polynomials while preserving invariance and positivity. This framework has been recently introduced for tensor decompositions, in particular for quantum many-body systems. Here we transfer results about decomposition structures, invariance under  permutations of variables, positivity, rank inequalities and separations, approximations, and undecidability  to real polynomials. Specifically, we define invariant decompositions of polynomials, and characterize which polynomials admit such decompositions. We then include positivity: We define invariant separable and sum-of-squares decompositions, and characterize the polynomials similarly. We provide inequalities and separations between the ranks of the decompositions, and show that the separations are not robust with respect to approximations. For cyclically invariant decompositions, we show that it is undecidable whether the polynomial is nonnegative or sum-of-squares for all system sizes. Our framework is different from existing approaches for polynomial decompositions, since it covers symmetry and positivity combined, in a clean and uniform way. Also, our work sheds new light on polynomials by putting them on an equal footing with tensors, and opens the door to extending this framework to other tensor product structures.  
\end{abstract}

\maketitle

\section{Introduction}

In a theory, the description of the elementary constituents is as important as the description of their composition. 
In quantum theory, a few postulates describe the behaviour of individual quantum systems, 
and one postulate describes how to compose them (mathematically, with the tensor product).  
Another example are multivariate polynomials, 
which can be constructed as the composition of the spaces of univariate polynomials with the tensor product. 
Both aspects are crucial---the elementary constituents and the composition---, and it is a misconception of reductionism to overestimate the importance of individual systems. 

The opposite of composing is \emph{decomposing}---expressing an object in terms of elementary constituents. 
This can be seen as an \emph{inverse problem} of the structure provided by the composition, and is generally a very rich problem. 
In many occasions, we want a decomposition that reflects the properties of the global object, that is, that provides a ``certificate'' of a global property in the local objects. 
For example, combining identical objects gives rise to a symmetric global object, 
or a sum of positive elementary constituents gives rise to a positive global object---the latter is particularly important in quantum theory, 
where \emph{entangled} objects are those \emph{not} admitting a  certain kind of positive decomposition. 
Which properties of the global object can be ``witnessed'' by the local objects? 
Answering this question amounts to solving the inverse problem, 
as it requires \emph{characterising which global properties can be transferred to the local objects, and how.} 

Recently, a framework to describe decompositions in tensor product spaces has been introduced \cite{De19d}, focusing on two aspects of this characterisation. 
The first is \emph{invariance}, namely,  if the global object is invariant under the exchange of some elementary constituents, can this be reflected in the decomposition? 
Ref.\ \cite{De19d} clarified what it means `to be reflected in the decomposition' by defining an `invariant decomposition', 
and gave  sufficient conditions for the transfer of invariance from the global to the local objects.
The second aspect is \emph{positivity}, namely, if the global object has some positivity property (is in some cone), 
can this be reflected in the decomposition? Ref.\ \cite{De19d} also studied this question, in combination with the invariance.  
In addition, this framework was extended to the approximate case, where the decomposition is content with  \emph{almost} realising the global object---often giving rise to big savings in the cost of the decomposition  \cite{De20}.

This framework is inspired by tensor decompositions---in particular, by the description of quantum many-body systems. 
Yet, it applies to all tensor product structures. In this paper, we apply it to \emph{real multivariate polynomials}. These are objects in the tensor product space of  polynomials in each of their variables,  
$$
\mathcal{P} \coloneqq \mathbb{R}[\mathbf{x}^{[0]},\mathbf{x}^{[1]}, \ldots, \mathbf{x}^{[n]}] \cong \mathbb{R}[\mathbf{x}^{[0]}] \otimes \mathbb{R}[\mathbf{x}^{[1]}] \otimes \cdots \otimes \mathbb{R}[\mathbf{x}^{[n]}]  , 
$$
where $\otimes$ denotes the algebraic tensor product and $\mathbf{x}^{[i]}$ a collection of variables $x^{[i]}_1, \ldots x^{[i]}_{m_i}$. In other words, every polynomial $p \in \mathcal{P}$ can be expressed as a finite sum of ``elementary constituents''
$$
p^{[0]}(\mathbf{x}^{[0]}) \cdot p^{[1]}(\mathbf{x}^{[1]}) \cdots p^{[n]}(\mathbf{x}^{[n]}),  
$$
where every $p^{[i]}$ is itself a polynomial that only depends on the variables $\mathbf{x}^{[i]}$.
We consider two questions:
\begin{enumerate}[label=(\alph*)]
\item If $p$ is symmetric under the exchange of, say, systems $[i]$ and $[j]$, can this symmetry be reflected in the decomposition? 

\item If $p$ is positive (for some notion of positivity), can this positivity be reflected in the decomposition? 
\end{enumerate}
Our framework solves these two questions in the following way---in particular applied to polynomials: 
\begin{enumerate}[label=(\alph*),ref=(\alph*)]
\item \label{a} 
The summation structure is described by a weighted simplicial complex $\Omega$, 
so that  every system  $i$ is associated to a vertex of $\Omega$,
and every summation index to a facet of $\Omega$. 

\item \label{b} 
By definition, an $(\Omega, G)$-decomposition of a polynomial contains a  certificate  of invariance under the group $G$. 
We characterise which $G$-invariant polynomials admit an $(\Omega,G)$-decomposition.

\item \label{c} 
By definition, a separable or sum-of-squares (sos) $(\Omega, G)$-decomposition contains a certificate of invariance and of membership in the separable or sos cone, respectively. 
We characterise which separable or sos polynomials admit such decompositions. 
\end{enumerate}

To be  specific, this framework is inspired by decompositions of quantum many-body systems provided by \emph{tensor networks} \cite{Or18}. 
The latter are prominent in quantum information theory and condensed matter physics (and recently machine learning), 
and favour certain arrangements of the summation indices---for example, the indices can be arranged in  a circle: 
\be 
\label{eq:circular} p= \sum_{\alpha_0, \ldots, \alpha_n = 1}^{r} p^{[0]}_{\alpha_0, \alpha_1}(\mathbf{x}^{[0]}) \cdot p^{[1]}_{\alpha_1, \alpha_2}(\mathbf{x}^{[1]}) \cdots p^{[n]}_{\alpha_n, \alpha_0}(\mathbf{x}^{[n]}). 
\ee
(This arrangement is motivated by the structure of physical interactions). 
Note that we have already written the previous equation for a polynomial $p$, as both quantum many-body systems and polynomials compose with the tensor product. 
From a mathematical perspective, the natural decomposition is the one with a single  index, namely
\be \label{eq:simplex}
p = \sum_{\alpha = 1}^{r} p^{[0]}_{\alpha}(\mathbf{x}^{[0]}) \cdot p^{[1]}_{\alpha}(\mathbf{x}^{[1]}) \cdots p^{[n]}_{\alpha}(\mathbf{x}^{[n]}).
\ee
In both cases, the smallest integer $r$ measures the cost of decomposing the polynomial---the one  of \eqref{eq:simplex} is called the \emph{tensor rank}. 
Our framework puts both decompositions under one umbrella: 
in Equation \eqref{eq:circular}, the weighted simplicial complex is the circle graph, 
and in \eqref{eq:simplex}, it is the full simplex (cf.\ \ref{a}).  

Symmetries are central in physics, both conceptually and practically, and it is  impossible to overstate their importance   in mathematics.  
Our framework models symmetries as follows: 
we have a group $G$ acting on the set $\{0, \ldots, n\}$, 
 and the induced action on the polynomial space $\mathcal{P}$ is obtained by permuting system $[i]$ to $[gi]$, 
$$
g: \mathbf{x}^{[i]} \mapsto g \mathbf{x}^{[i]} \coloneqq \mathbf{x}^{[gi]}. 
$$
A polynomial is $G$-invariant if it is invariant with respect to all such permutations $ g\in G$, 
and we want to make this invariance   explicit in the decomposition of $p$. 
For example, the following decomposition 
$$
p= \sum_{\alpha_0, \ldots, \alpha_n = 1}^{r} p_{\alpha_0, \alpha_1}(\mathbf{x}^{[0]}) \cdot p_{\alpha_1, \alpha_2}(\mathbf{x}^{[1]}) \cdots p_{\alpha_n, \alpha_0}(\mathbf{x}^{[n]})  
$$
makes explicit that $p$ is invariant under the cyclic group, $\mathbf{x}^{[i]} \mapsto \mathbf{x}^{[i+1]}$. 
(Note that there are no superscripts $[i]$ in contrast to Equation \eqref{eq:circular}).  
And 
$$
p= \sum_{\alpha = 1}^{r} p_{\alpha}(\mathbf{x}^{[0]}) \cdot p_{\alpha}(\mathbf{x}^{[1]}) \cdots p_{\alpha}(\mathbf{x}^{[n]}) 
$$
makes explicit that $p$ is invariant under the full symmetry group. 
The former is known in quantum physics as  the \emph{translationally invariant matrix product operator} form 
(and the minimal number $r$ as the \emph{t.i.\ operator Schmidt rank} \cite{De19}), 
and the latter as the \emph{symmetric tensor decomposition} (and the minimal $r$ as   the \emph{symmetric tensor rank}  \cite{Co08c,Sh17b}). 
In our framework, the former corresponds to the circle with the cyclic group, and the latter to the full simplex with the full permutation group (cf.\ \ref{b}).

Finally, if $p$ is in a cone (such as sum-of-squares (sos) polynomials or the cone of nonnegative polynomials), 
we want a  certificate  of this fact (cf.\ \ref{c}). 
In quantum physics, a quantum state is positive semidefinite and the certificate  is called a purification. 
In probabilistic modelling, the certificate of a probability distribution  is a nonnegative decomposition. 
In real algebraic geometry, the natural certificate of positivity of a polynomial is being sum of squares. 
In all of these cases,  witnessing the positivity of a global element is a central problem with many ramifications. 

Note that decompositions of tensors and polynomials  have already been studied a lot from different perspectives. Also symmetries and positivity have been  considered combined, but the arising decomposition are by far not as clean as the separate decompositions. To give a short overview, and thus also motivate our combined approach, let us explain some of the existing decompositions, and point out why they are not directly related to our approach.

The Waring decomposition is a decomposition of polynomials, which is also inspired by tensors. Let $p \in \mathbb{R}[x_1, \ldots, x_n]$ of degree $d$. The \emph{Waring rank} of $p$ is defined as the minimum $r \in \mathbb{N}$ such that
	$$ p = \sum_{\alpha = 1}^{r} c_\alpha \ell_{\alpha}(x_1, \ldots, x_n)^d$$
	where $\ell_\alpha(x_1, \ldots, x_n) = a_{\alpha,1} x_1 + \ldots + a_{\alpha,n} x_n$ is a linear form. The Waring rank is equivalent to the symmetric tensor rank by applying the correspondence
	$$ p = \sum_{i_1, \ldots, i_n = 1}^{d} T_{i_1, \ldots, i_n} x_{i_1} \cdots x_{i_n}$$
	between symmetric tensors in $T \in \left(\mathbb{C}^{d}\right)^{\otimes n}$ and homogeneous polynomials of degree $n$. 	
	Yet, the Waring decomposition cannot exhibit any additional symmetry of the polynomial, since the corresponding tensor is already fully symmetric for any polynomial. For generalizations of the Waring problem to polynomials instead of linear forms, we refer to \cite{Fr12}. Another related decomposition is the completely decomposable decomposition \cite{Ab14}.

For symmetric polynomials, the decomposition into power-sum polynomials is an example of an explicitly invariant decomposition. Every symmetric polynomial $p$ can be written as $p = q(p_1, \ldots, p_n)$, where
	$$p_{\alpha} = \sum_{i=1}^{n} x_i^{\alpha}.$$
	In other words, the ring of symmetric polynomials with real coefficients corresponds to the ring $\mathbb{R}[p_1, \ldots, p_n]$ generated by power-sum polynomials.
The same statement is true by replacing the set of power-sum polynomials by elementary symmetric polynomials.

Also, the combination of symmetry and positivity is well-studied. It is, for example, known that symmetric sum-of-squares polynomials do, in general, not decompose into a sum of symmetric squares, to fully characterize the set of symmetric sum-of-squares polynomials, one has to introduce a more general notion of symmetric sum-of-square decomposition \cite{De21d}.

In this paper we do the following: 
\begin{enumerate}[label=(\roman*)]
\item 
We define invariant decompositions of polynomials (Definition \ref{def:omegaG-dec}). 
We show that  every invariant polynomial admits an invariant decomposition if the group action is free on the weighted simplicial complex (Theorem \ref{thm:omegaG-dec}),  
and that every group action can be made free by increasing the number of summation indices (Proposition \ref{pro:extend}).  
In addition, every invariant polynomial can be written as the difference of two invariant decompositions 
if the group action is blending (Theorem \ref{thm:omegaG-diff}).

\item 
We  define  the invariant separable  decomposition  (Definition \ref{def:sepdec}), 
and the invariant sos decomposition (Definition \ref{def:sos-dec}), and  show that every invariant separable/sos polynomial admits an invariant separable/sos decomposition if the group action is free (Theorem \ref{thm:sep} and Corollary \ref{cor:sos}, respectively). These decompositions combine positivity and symmetry in a very clean way.

\item 
We provide inequalities and separations between the ranks of three invariant decompositions (Proposition \ref{pro:inequalitites} and Corollary \ref{cor:rank-sep}, respectively). 

\item  We show that the separations are not robust with respect to approximations (Theorem \ref{thm:approxSeparable}).

\item  For decompositions on the circle with translational invariance, 
we show that it is undecidable whether the global polynomial is sos or nonnegative for all system sizes (Theorem \ref{thm:und}).
\end{enumerate}
Throughout this work, an  `invariant decomposition' refers to an $(\Omega,G)$-decomposition,
and an `invariant polynomial' to a $G$-invariant polynomial. 
Similarly, an `invariant separable/sos decomposition' refers to a separable/sos $(\Omega,G)$-decomposition.

This paper is organized as follows. 
In Section \ref{sec:wsc} we define weighted simplicial complexes  and group actions. 
In Section \ref{sec:dec} we define and study the invariant decomposition, 
the invariant separable decomposition and the invariant sum of squares decomposition. 
In Section \ref{sec:inequalitites} we study inequalities and separations between the ranks. 
In Section \ref{sec:approx} we study the approximate case. 
In Section \ref{sec:undecidable} we show that a problem related to positive polynomials is undecidable. 
In Section \ref{sec:concl} we conclude and provide an outlook.

\section{Weighted simplicial complexes and group actions}
\label{sec:wsc}

Here we define weighted simplicial complexes (Section \ref{ssec:wsc}) and groups acting on them (Section \ref{ssec:actions}), both defined in \cite{De19d}. These constitute the underlying topological structure on which we will consider invariant polynomial decompositions. 

Throughout this paper, we use the notation $[n] \coloneqq \{0,\ldots, n\}$, and  denote its power set $\mathcal P([n])$ by $\mathcal P_n$. 

\subsection{Weighted simplicial complexes}
\label{ssec:wsc}

We now define weighted simplicial complexes and refer to \cite{Da90} for details. 
Examples of weighted simplicial complexes are given in Section \ref{ssec:actions} and in \cite{De19d}.

\begin{definition}[Weighted simplicial complexes]\label{def:wsc} \quad\nobreakpar

\begin{enumerate}[label=(\roman*),ref=(\roman*),leftmargin=*]
\item \label{def:wsc:i}
A \emph{weighted simplicial complex} on $[n]$ is a map 
$$\Omega\colon \mathcal{P}_n\to \mathbb{N}$$ such that $\Omega(S_1)$ divides $\Omega(S_2)$ whenever $S_1 \subseteq S_2$. $\Omega$ is called a \emph{simplicial complex} if $\Omega(\mathcal{P}_n) = \{0,1\}$.
 
\item 
A set $S \in \mathcal{P}_n$ is called a \emph{simplex} of $\Omega$ if $\Omega(S) \neq 0$. We will always assume that each singleton $\{i\}$ is a simplex, and call the elements $i\in[n]$ the \emph{vertices} of the weighted simplicial complex.
We call a maximal simplex (with respect to inclusion) a \emph{facet} of $\Omega.$
Moreover, we denote the collection of all facets by 
$$\mathcal{F} \coloneqq \left\{ F\in\mathcal{P}_n :  F \textrm{ facet of } \Omega \right\},$$
and for each vertex $i$ the collection of facets that contain $i$ by
$$\mathcal{F}_i \coloneqq \left\{F \in \mathcal{F} : i\in F\right\}.$$
By restricting $\Omega$ to $\mathcal{F}$ or $\mathcal{F}_i$ we can interpret these mappings as multisets which we call $$\widetilde{\mathcal F} \ \textrm{ and  } \ \widetilde{\mathcal{F}}_i.$$  $\widetilde{\mathcal{F}}$ contains each facet $F$ exactly $\Omega(F)$-many times. Moreover, we introduce the canonical \emph{collapse map} $$c \colon \widetilde{\mathcal F}\to \mathcal F, \quad c\colon \widetilde{\mathcal F}_i\to \mathcal F_i,$$ mapping all copies to the underlying facet.

\item 
Two vertices $i,j$ are \emph{neighbours} if 
$$
\mathcal{F}_i \cap \mathcal{F}_j \neq \emptyset \qquad (\textrm{or equivalently if } \widetilde{\mathcal{F}}_i \cap \widetilde{\mathcal{F}}_j\neq\emptyset).
$$  
Two vertices are \emph{connected} if there exists a sequence of neighbours $i_0, \ldots, i_k$ such that $i = i_0$ and $j = i_k$. 
We say that the weighted simplicial complex is connected if every pair of vertices is connected. 
\end{enumerate}
\end{definition}

Note that a simplicial complex $\Omega$ is  the characteristic function of a subset $\mathcal{A} \subseteq \mathcal{P}_n$. By definition of $\Omega$, $\mathcal{A}$ is closed under passing to subsets. This is the usual definition of an (abstract) simplicial complex. 

Note also that a weighted simplicial complex is a special case of a multihypergraph \cite{Bo98}, in the sense that all simplices of a facet are included, and in addition  the multiplicities satisfy Definition \ref{def:wsc} \ref{def:wsc:i}. 
Our framework could also be formulated with multihypergraphs, as  the  decompositions only depend on the multifacets $\widetilde{\mathcal{F}}$. Nonetheless, we find  the slightly less general notion of a weighted simplicial complex more convenient to apply to this framework.

In the following we introduce two basic examples---the single and the double edge---which will serve as a running example throughout the paper.

\begin{example}[The simple and double edge] \label{ex:wscbasic} \quad\nobreakpar
\begin{enumerate}[label=(\roman*),leftmargin=*]
\item 
Consider two vertices and the weighted simplicial complex $\Omega = \Lambda_1$ which maps every subset of $\{0,1\}$ to $1$. This is just the simple edge, consisting of exactly one (multi)-facet $\mathcal{F} = \big\{\{0,1\}\big\}$.  

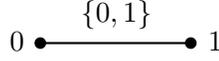
\vspace{0.6cm}
\begin{center}
\begin{tikzpicture}
\filldraw (-1,0) circle (2pt);\put (-40,-3){$0$};
\filldraw (1,0) circle (2pt);\put (35,-3){$1$};
\draw[thick] (-1,0) -- (1,0); \put (-13,8){$\{0,1\}$};
\end{tikzpicture}
\end{center}

\bigskip\noindent

\item  Adding a second facet, we obtain the \emph{double edge} $\Delta$, which is the weighted simplicial complex on $\mathcal P_1$ that assigns the value $1$ to the sets $\{0\}, \{1\}$ and the value $2$ to $\{0,1\}.$

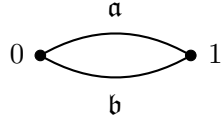
\vspace{0.6cm}
\begin{center}
\begin{tikzpicture}
\filldraw (-1,0) circle (2pt);\put (-40,-3){$0$};
\filldraw (1,0) circle (2pt);\put (35,-3){$1$};
\draw[thick] (-1,0) to[out=-30, in=210] (1,0); \put (-3,15){$\mathfrak{a}$};
\draw[thick] (-1,0) to[out=30, in=150] (1,0);\put (-3,-22){$\mathfrak{b}$};
\end{tikzpicture}
\end{center}

\bigskip\noindent
In this case $\mathcal{F}$ differs from $\widetilde{\mathcal{F}}$ since
$$\mathcal{F}_0=\mathcal{F}_1=\mathcal F=\big\{\{0,1\}\big\}$$ are singletons, but  $$\widetilde{\mathcal{F}}_0=\widetilde{\mathcal{F}}_1=\widetilde{\mathcal{F}}=\{\mathfrak{a},\mathfrak{b}\}$$ are not.
\end{enumerate}
\demo \end{example}

\subsection{Group actions}\label{ssec:actions}

We now introduce group actions on the set $[n]$, and promote them to actions on weighted simplicial complexes. 
For the reader not familiar  with group actions on sets, we refer to  \cite{La02}. 
Throughout this paper, we denote the identity element of a group $G$  by $e$.

\begin{definition}[Group actions]\label{def:action} \quad\nobreakpar
\begin{enumerate}[label=(\roman*), ref=(\roman*),leftmargin=*]
\item \label{im:action:i}
Let $G$ be a group acting on the sets $X$ and $Y$, respectively. A map $f\colon X\to Y$  is called  \emph{$G$-linear} if 
\be \nn f(gx)=gf(x)
\ee
holds for all $x\in X,g\in G$. 
If $G$ acts trivially on $Y$ (i.e.\ $gy = y$ for all $g \in G$ and $y \in Y$), we instead call $f$ \emph{$G$-invariant}. 

\item \label{im:action:ii}
If $G$ acts on $X$,   for any map $f \colon X\to Y$ and any $g\in G$ we define a new map 
\be \nn {}^g f \colon X \to Y: x \mapsto f(g^{-1}x).\ee 
It is immediate that
$$
{}^h\left({}^gf\right)= {}^{hg} f\: \textrm{ and }\: {}^e f=f,
$$ 
so this defines an action of $G$ on the set of all maps from $X$ to $Y$.
In particular, the function $f\mapsto{}^gf$ is a bijection on this set. 
If $f$ is defined only on a subset $A\subseteq X$, then ${}^gf$ acts on the translated subset $$gA\coloneqq \{ gx : x\in A\}\subseteq X.$$

\item \label{im:action:iii}
An action of $G$ on $X$ is called \emph{free} if all its stabilizers are trivial, i.e.\  ${\rm Stab}(x)=\{e\}$ for every $x\in X$, where $${\mathrm{Stab}}(x) \coloneqq \left\{ g\in G : gx=x\right\}.$$

\item \label{im:action:iv}
We call an action $G$ on $[n]$ \emph{blending}  if $\{g_0 0,\ldots,g_n n\}=[n]$ for certain $g_0,\ldots, g_n\in G$ implies the existence of $g\in G$ with $g i=g_i i$ for all  $i=0,\ldots, n$. In words, a permutation of $[n]$ given by different group elements can also be achieved by a single group element.
\end{enumerate}
\end{definition}

We now promote a group action on $[n]$ to a group action on a weighted simplicial complex: 

\begin{definition}[Group action on a weighted simplicial complex]
\quad\nobreakpar
\begin{enumerate}[label=(\roman*),leftmargin=*]
\item  A \emph{group action of $G$ on the weighted simplicial complex $\Omega$} consists of the following:
\begin{enumerate}[label=(\alph*),labelindent=0pt]
 \item 
 A group action of $G$ on $[n]$  such that the map  $\Omega$ is $G$-invariant with respect to the canonical action of $G$ on $\mathcal{P}_n$ (i.e.\ it permutes vertices in a way that simplices become simplices of the same weight). This induces a well-defined action of $G$ on $\mathcal{F}$. 
 \item 
 An action of $G$ on the set of multifacets $\widetilde{\mathcal{F}}$ such that the canonical collapse map $$c\colon{\widetilde{\mathcal{F}}}\to\mathcal{F}$$ is $G$-linear. 
 The action $G$ on $\widetilde{\mathcal{F}}$ is a \emph{refinement} of the action of $G$ on $\mathcal{F}$. 
  \end{enumerate}

\item We call the action $G$ on the weighted simplicial complex $\Omega$ \emph{free} if the action of $G$ on $\widetilde{\mathcal{F}}$ is free.
\end{enumerate}
\end{definition}

\begin{remark}[Group actions]\label{rem:grac}\quad\nobreakpar
\begin{enumerate}[label=(\roman*),ref=(\roman*),leftmargin=*]
\item \label{rem:grac:i}
Since every weighted simplicial complex consists of finitely many vertices, we will assume the group $G$ to be finite as well. We could also assume that $G$ is a subgroup of the permutation group $S_{n+1}$ (since every group action can be understood as a collection of permutations on $[n]$), but sometimes it is more convenient not to choose the latter representation.

\item \label{rem:grac:ii}
A group action on a weighted simplicial complex $\Omega$ permutes the vertices $[n]$ in a way that preserves the structure of the complex. 
In particular, it induces an action of $G$ on $\mathcal{F}$, where all facets in the same orbit are of the same weight. 
Note that each $g\in G$  provides a weight-preserving bijection 
\begin{align*} 
 g \colon \mathcal{F}_i & \to \mathcal{F}_{gi} \\ 
 F& \mapsto gF.
\end{align*}

\item  \label{rem:grac:iii}
To obtain a group action on a weighted simplicial complex  with multifacets one needs to provide additional information, namely how elements $g \in G$ permute the different copies of facets when mapping a facet $F$ to $gF$. Obviously, any group action can be refined, but there are many ways of doing so.

\item  \label{rem:grac:iv}
The notion of a blending group action (on a weighted simplicial complex) just refers to the action of $G$ on the vertices $[n]$.

\item  \label{rem:grac:v}
The notion of a free group action  on a weighted simplicial complex always concerns the action of $G$ on $\widetilde{\mathcal{F}}$. 
The action of $G$ on the vertices can be free without the action of $G$ on $\Omega$ being free (see Example \ref{ex:wsc}). On the other hand, an action of $G$ on $\Omega$ can be free without the   action of $G$ on the underlying vertices $[n]$ or on the facets $\mathcal{F}$ being free.
As we will see in Proposition \ref{pro:extend}, any action of $G$ on $\Omega$ can be refined to a free group action, after enlarging the weights of the facets. This, combined with Theorem \ref{thm:omegaG-dec},  justifies our choice of \emph{weighted} simplicial complexes in our framework. 

\item  \label{rem:grac:vi}
An action of $G$ on a set $X$ is free if and only if there exists a $G$-linear map 
$${\textbf{z}}\colon X \to G$$
where $G$ acts on itself via left-multiplication (which is obviously free). To construct $\textbf{z}$ for a free action, choose for each orbit an element $x$ and map $gx$ to $g$. The reverse implication is immediate.
\demo\end{enumerate}
\end{remark}

Let us now discuss the  group actions  on the simple and double edge of Example \ref{ex:wscbasic}.

\begin{example}[The simple and double edge with group actions]\label{ex:group_EdgeDoubleEdge} \quad\nobreakpar
\begin{enumerate}[label=(\roman*),leftmargin=*]
\item For the simple edge $\Lambda_1$  there is only one interesting group action, namely by $C_2=S_2$, which permutes the vertices $0,1$. Although this group action is free and blending on $\{0,1\}$, it is not free on the weighted simplicial complex, since the (only) facet remains fixed under each group element.

\item For the double edge $\Delta$  the group action of $C_2$ can be extended to the multifacets in two different ways. One extension keeps each multifacet fixed, in which case the action is not free, 
and the other  one permutes the multifacets, i.e.\ 
flips $\mathfrak{a}$ and $\mathfrak{b}$, in which case the action is free.
Henceforth, when we refer to $C_2$ on $\Delta$ we always refer to the free refinement. \demo
\end{enumerate}
\end{example}

There are  other canonical examples of weighted simplicial complexes and group actions  which will play a   role in the development of invariant polynomial decompositions. Let us introduce  them now. 

\begin{example}[The simplex, the line and the circle]\label{ex:wsc}\quad\nobreakpar
\begin{enumerate}[label=(\roman*), ref=(\roman*),leftmargin=*]
\item \label{ex:wsc:i}
The simplicial complex $\Omega = \Sigma_n$ mapping each subset of $[n]$ to $1$ is called the \emph{$n$-simplex}. For $n=4$ it can be depicted as 
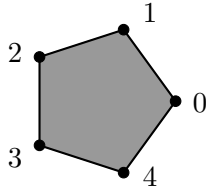
\vspace{0.4cm}
\begin{center}
\begin{tikzpicture}
\filldraw (1,0) circle (2pt);
\filldraw (0.309,0.951) circle (2pt);
\filldraw (-0.809,0.588) circle (2pt);
\filldraw (-0.809,-0.588) circle (2pt);
\filldraw (0.309,-0.951) circle (2pt);

\draw[thick] (1,0) -- (0.309,0.951) -- (-0.809,0.588) -- (-0.809,-0.588) -- (0.309,-0.951) -- (1,0);
\filldraw[opacity=0.4] (1,0) -- (0.309,0.951) -- (-0.809,0.588) -- (-0.809,-0.588) -- (0.309,-0.951) -- (1,0);;
\put (35,-4) {$0$};
\put (16,30) {$1$};
\put (-35,15) {$2$};
\put (-35,-25) {$3$};
\put (16,-32) {$4$};
\end{tikzpicture}
\end{center}
\vspace{0.2cm}\noindent
where it contains only one facet,   $\mathcal{F}=\widetilde{\mathcal{F}}=\{[n]\}$. 
Any group action on $[n]$ is   a group action on $\Sigma_n$, but it clearly is the trivial group action on $\widetilde{\mathcal{F}}$. The action of the full permutation group $S_{n+1}$ (which contains $(n+1)!$ elements) on $[n]$ is blending. 
The only free action on $\Sigma_n$ is the action from the trivial group. However, if the weight of the only facet is enlarged to $\vert G\vert$, any action from $G$ on $[n]$ has at least one free refinement by Proposition \ref{pro:extend}.

\item  \label{ex:wsc:ii}
For $n\geq 1$, the \emph{line of length $n$} is the simplicial complex $\Omega=  \Lambda_n$ given by the following graph:

\bigskip
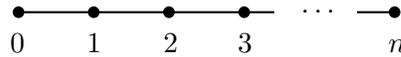
\begin{center}
\begin{tikzpicture}
\filldraw (-2,0) circle (2pt);\put (-60,-15){$0$};
\filldraw (-1,0) circle (2pt);\put (-31,-15){$1$};
\filldraw (0,0) circle (2pt);\put (-2,-15){$2$};
\filldraw (1,0) circle (2pt);\put (26,-15){$3$};
\put (50,-2.5) {$\cdots$};
\filldraw (3,0) circle (2pt);\put (83,-15){$n$};
\draw[thick] (-2,0) -- (-1,0);
\draw[thick] (-1,0) -- (0,0);
\draw[thick] (0,0) -- (1,0);
\draw[thick] (1,0) -- (1.4,0);
\draw[thick] (2.5,0) -- (3,0);
\end{tikzpicture}
\end{center}

\vspace{0.7cm}
\noindent The collection of facets $\mathcal{F} = \widetilde{\mathcal{F}}$ consists of $n$ elements.
The only non-trivial group action on $\Lambda_n$ is given by the cyclic group with two elements $G = C_2$, where the generator inverts the order of the vertices, i.e.\ vertex $i$ is sent to vertex $n-i$. 
This action is free if and only if $n$ is even, and blending if and only if  $n \leq 2$. If $n$ is odd, the action admits a free refinement if the weight of the middle edge is increased to $2$. For $n = 1$ we regain the single edge.

\item  \label{ex:wsc:iii}
For $n\geq 3$, the \emph{circle of length $n$} is the simplicial complex $\Omega = \Theta_n$  corresponding to  the following graph:

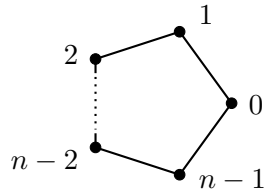
\vspace{0.4cm}
\begin{center}
\begin{tikzpicture}
\filldraw (1,0) circle (2pt);
\filldraw (0.309,0.951) circle (2pt);
\filldraw (-0.809,0.588) circle (2pt);
\filldraw (-0.809,-0.588) circle (2pt);
\filldraw (0.309,-0.951) circle (2pt);

\draw[thick] (1,0) -- (0.309,0.951) -- (-0.809,0.588);
\draw[thick] (-0.809,0.588) -- (-0.809, 0.4);
\draw[thick, dotted] (-0.809, 0.35) -- (-0.809, -0.35);
\draw[thick] (-0.809,-0.588) -- (-0.809, -0.4);
\draw[thick] (-0.809,-0.588) -- (0.309,-0.951) -- (1,0);
\put (35,-4) {$0$};
\put (16,30) {$1$};
\put (-35,15) {$2$};
\put (-55,-25) {$n-2$};
\put (16,-32) {$n-1$};
\end{tikzpicture}
\end{center}
\vspace{0.2cm}

\noindent which has $n$ facets.  A canonical action is given by the cyclic group $G=C_{n}$, which is generated by translation of the vertex $i \mapsto i + 1 \textrm{ mod } (n)$. This action is free on $\Omega$ but not blending.\demo
\end{enumerate}
\end{example}

We now state what we have already seen in Example \ref{ex:wsc} \ref{ex:wsc:i}, \ref{ex:wsc:ii} and \ref{ex:wsc:iii}  in a more general setting, namely that  by increasing the multiplicity of facets of a weighted simplicial complex $\Omega$ we can make every group action free. In short, \emph{every group action has a free refinement}. 
It is good to bear this in mind for the rest of the paper, 
because we will need to assume freeness in many results, 
but this is a ``mild'' assumption because of Proposition \ref{pro:extend}. 
This proposition is proven in \cite{De19d}.

\begin{proposition}[Free refinement \cite{De19d}]\label{pro:extend}
Every action of a finite group $G$ on a connected weighted simplicial complex $\Omega$ has a free refinement, 
which in particular can be obtained by multiplying the weight of every facet of $\Omega$ by $\vert G\vert$. 
\end{proposition}

\section{Invariant polynomial decompositions and ranks}
\label{sec:dec}

In this section we define invariant polynomial decompositions and their ranks. 
To this end we first set the stage (Section \ref{ssec:stage}), and then 
define and study the invariant decomposition (Section \ref{ssec:invdec}), 
the invariant separable decomposition (Section \ref{ssec:sepdec}), 
and finally the invariant sum-of-squares decomposition (Section \ref{ssec:sosdec}). 

\subsection{Setting the stage}\label{ssec:stage}
Throughout this section we consider polynomials in the space
$$
\mathcal{P} \coloneqq \mathbb{R}[\mathbf{x}^{[0]}, \mathbf{x}^{[1]}, \ldots, \mathbf{x}^{[n]}] \cong \mathbb{R}[\mathbf{x}^{[0]}] \otimes \mathbb{R}[\mathbf{x}^{[1]}] \otimes \cdots \otimes \mathbb{R}[\mathbf{x}^{[n]}]
$$
where $\mathbb{R}[\mathbf{x}^{[i]}] \coloneqq \mathbb{R}[x^{[i]}_1, \ldots, x^{[i]}_{m_i}]$ is the space of real polynomials in $m_i$ variables, and $\otimes$ denotes  the algebraic tensor product.
These polynomials use collections of local variables, denoted $\mathbf{x}^{[i]}$, for each local site $i=0,\ldots, n$. 
The case where all $m_i=1$ is already very interesting, as it describes how the multivariate polynomial ring is decomposed into a tensor product of univariate polynomial rings.

In particular, $\mathbb{R}[x^{[0]}, \ldots, x^{[n]}] \cong \mathbb{R}[x^{[0]}] \otimes \mathbb{R}[x^{[1]}] \otimes \cdots \otimes \mathbb{R}[x^{[n]}]$, where $x^{[i]}$ is a single variable, means that every multivariate polynomial can be expressed as a sum of products of uni-variate polynomials, i.e.\
$$ p = \sum_{\alpha = 1}^{r} p_{\alpha}^{[0]}(x^{[0]}) \cdots p_{\alpha}^{[n]}(x^{[n]}).$$

We define the local degree of $p \in \mathcal{P}$, denoted $\deg_{\mathrm{loc}}(p)$, as the smallest positive integer $d \in \mathbb{N}$ such that
$$
p \in \mathcal{P}_d \coloneqq \mathbb{R}[\mathbf{x}^{[0]}]_d \otimes \mathbb{R}[\mathbf{x}^{[1]}]_d \otimes \cdots \otimes \mathbb{R}[\mathbf{x}^{[n]}]_d
$$
where $\mathbb{R}[\mathbf{x}]_d$ is the space of real polynomials in $\mathbf{x}$ of degree at most $d$. A polynomial with $\deg_{\mathrm{loc}}(p) \leq d$ contains monomials consisting of variables in $\mathbf{x}^{[i]}$ with degree at most $d$, for each $i$.
Note that the local degree can be related with the (global) degree of the polynomial by
$$
\deg_{\mathrm{loc}}(p) \leq \deg(p) \leq (n+1) \cdot \deg_{\mathrm{loc}}(p).
$$

A given group action $G$ on $[n]$ also induces a group action on the space $\mathcal{P}$. 
The action is defined for $g \in G$ and $p \in \mathcal{P}$ by
\be\label{eq:gp}
(gp)(\mathbf{x}^{[0]}, \ldots, \mathbf{x}^{[n]}) \coloneqq p(\mathbf{x}^{[g0]}, \ldots, \mathbf{x}^{[gn]}).
\ee
Note that this definition only makes sense if the local polynomial spaces $\mathbb{R}[\mathbf{x}^{[i]}]$ and $\mathbb{R}[\mathbf{x}^{[j]}]$ are isomorphic whenever $i,j \in [n]$ are in the same orbit of $G$ (i.e.\ $gi = j$ for some $g \in G$), i.e.\ the number of local variables needs to coincide for $i,j$, namely $m_i = m_j$. 
The canonical isomorphism between elements in $\mathbb{R}[\mathbf{x}^{[i]}]$ and $\mathbb{R}[\mathbf{x}^{[j]}]$ is  given by replacing the variables $\mathbf{x}^{[i]}$ with $\mathbf{x}^{[j]}$ in every polynomial and vice versa. 
We will  frequently  use this isomorphism in an implicit way, 
as for a polynomial $p^{[i]} \in \mathbb{R}[\mathbf{x}^{[i]}]$ we will denote its corresponding element in $\mathbb{R}[\mathbf{x}^{[j]}]$ as $p^{[i]}(\mathbf{x}^{[j]})$.

We say that $p \in \mathcal{P}$ is \emph{$G$-invariant} if for each $g \in G$ we have $gp = p$, or equivalently
$$
p(\mathbf{x}^{[g0]}, \ldots, \mathbf{x}^{[gn]}) = p(\mathbf{x}^{[0]}, \ldots, \mathbf{x}^{[n]}) \quad \textrm{for every $g \in G$}.
$$
For example, if $m_i = 1$ and $G$ is the full permutation group on $[n]$, then a polynomial $p$ is invariant if
$$p(x^{[0]}, \ldots, x^{[n]}) = p(x^{[\sigma(0)]}, \ldots, x^{[\sigma(n)]})$$
for every permutation $\sigma: [n] \to [n]$ which means that $p$ is invariant with respect to arbitrary permutations of variables.

For two sets $A,B$ we denote the set of all functions from $A$ to $B$ by $B^A$. 
If the set $A$ is finite, such functions are sometimes written as an $\vert A\vert$-tuple of values in $B$.  
In our case, we will consider $\mathcal{I}$ to be a finite index set,  
and sometimes write a map 
$\alpha \in \widetilde{\mathcal{F}} \to  \mathcal{I}$ as a tuple $\alpha \in \mathcal{I}^{\widetilde{\mathcal{F}}}$ 
with entries from $\mathcal{I}$ and  where the entries are indexed by the facets in $\mathcal{\widetilde{F}}$. 
If we have a function $\alpha: \widetilde{\mathcal{F}} \to \mathcal{I}$ 
and want to restrict its domain to  $\widetilde{\mathcal{F}}_i$ (for some index $i \in [n]$), in the tuple notation we write
$$ \alpha_{\vert_i} \coloneqq \alpha_{\vert_{\mathcal{\widetilde{F}}_i}} \in \mathcal{I}^{\widetilde{\mathcal{F}}_i}, $$
which means that we delete all entries which are indexed by a facet not containing $i$.
We will in general stick to the functional notation except for the examples, where we will switch to the tuple notation. 
Their connection will be made explicit in Example \ref{ex:omegaG-dec}.

\subsection{The invariant decomposition}
\label{ssec:invdec}

We now define the  basic invariant decomposition, 
called $(\Omega,G)$-decomposition, 
simply called the \emph{invariant decomposition}. 
Afterwards we will study the existence of decompositions without invariance (page \pageref{sssec:withoutinv}),  
the existence of invariant decompositions with free group actions (page \pageref{sssec:free}) 
and with blending group actions (page \pageref{sssec:blending}).

The  idea of the invariant decomposition is to consider finite sums of elementary polynomials (i.e.\ polynomials written as a product of local polynomials  depending on one collection of variables $\mathbf{x}^{[i]}$), 
where each local polynomial is associated to a vertex of $\Omega$, 
and the summation indices are described as functions $\alpha_{\vert_i}$ on the facets. 
The following definition is illustrated in Example \ref{ex:omegaG-dec0}, 
\ref{ex:omegaG-exp} and \ref{ex:omegaG-dec}.

\begin{definition}[Invariant decomposition] \label{def:omegaG-dec} \quad\nobreakpar
\begin{enumerate}[label=(\roman*),ref=(\roman*),leftmargin=*]
\item \label{def:omegaG-dec:i}
An \emph{$(\Omega,G)$-decomposition} of $p \in \mathcal{P}$ consists of a finite index set $\mathcal{I}$ and families of polynomials
$$\mathcal{P}^{[i]} \coloneqq \left(p_{\beta}^{[i]} \right)_{\beta \in \mathcal{I}^{\widetilde{\mathcal{F}}_i}}$$
where $p_{\beta}^{[i]} \in \mathbb{R}\left[\mathbf{x}^{[i]}\right]$ for all $i \in [n]$, such that
\begin{enumerate}[label=(\alph*),ref=(\alph*),leftmargin=*]
	\item \label{im:a}
	$p$ can be written as
$$ p = \sum_{\alpha \in \mathcal{I}^{\widetilde{\mathcal{F}}}} p_{{\alpha_{\vert_0}}}^{[0]}(\mathbf{x}^{[0]}) \cdots p_{{\alpha_{\vert_n}}}^{[n]}(\mathbf{x}^{[n]})$$
	\item \label{im:b}
	For all $i \in [n]$, $g \in G$ and $\beta \in \mathcal{I}^{\widetilde{\mathcal{F}}_i}$ we have
$$p_{\beta}^{[i]}(\mathbf{x}^{[i]}) = p_{{}^g \beta}^{[gi]}(\mathbf{x}^{[i]})$$
where ${}^g \beta$ is defined in Definition \ref{def:action} ($ii$).
\end{enumerate}

\item  \label{def:omegaG-dec:ii}
The minimal cardinality of $\mathcal{I}$ among all $(\Omega,G)$-decomposition of $p$ is called the \emph{$(\Omega,G)$-rank of $p$}, denoted $ \rank_{(\Omega,G)}(p).$
If $p$ does not admit an $(\Omega,G)$-decomposition, we set $\rank_{(\Omega,G)}(p) = \infty$.

\item  \label{def:omegaG-dec:iii}
If $G$ is the trivial group action, we call the $(\Omega,G)$-decomposition just \emph{$\Omega$-decomposition} and denote its rank by \emph{$\rank_{\Omega}.$}  
\end{enumerate}
\end{definition}

Condition \ref{def:omegaG-dec:i} \ref{im:a} provides an  arrangement of the summation indices  encoded in the functions $\alpha$, and  
condition \ref{def:omegaG-dec:i} \ref{im:b} ensures that the decomposition has the desired symmetry, by requiring that the coefficients of particular local polynomials in different local spaces coincide. 
Note again that this equality only makes sense if the collections $\mathbf{x}^{[i]}$ and $\mathbf{x}^{[gi]}$  have the same cardinality (i.e.\ $m_i = m_{gi}$).

\begin{remark}[Admitting an $(\Omega,G)$-decomposition implies being $G$-invariant]\label{rem:OmegaG-dec} \quad\nobreakpar
\begin{enumerate}[label=(\roman*),ref=(\roman*),leftmargin=*]
\item \label{im:rem:OmegaG-dec:i}
If a polynomial has a $(\Omega,G)$-decomposition then it is $G$-invariant: 
\be \nn 
gp&=& p(\mathbf{x}^{[g0]}, \ldots, \mathbf{x}^{[gn]}) = \sum_{\alpha \in \mathcal{I}^{\widetilde{\mathcal{F}}}} p^{[0]}_{\alpha_{\vert_{0}}}(\mathbf{x}^{[g0]}) \cdots p^{[n]}_{\alpha_{\vert_{n}}}(\mathbf{x}^{[gn]}) \\ \nn &=& \sum_{\alpha \in \mathcal{I}^{\widetilde{\mathcal{F}}}} p^{[g0]}_{{}^g (\alpha_{\vert_{0}})}(\mathbf{x}^{[g0]}) \cdots p^{[gn]}_{{}^g (\alpha_{\vert_{n}})}(\mathbf{x}^{[gn]})  \\ \nn &=& \sum_{\alpha \in \mathcal{I}^{\widetilde{\mathcal{F}}}} p^{[g0]}_{({}^g\alpha)_{\vert_{g0}}}(\mathbf{x}^{[g0]}) \cdots p^{[gn]}_{({}^g\alpha)_{\vert_{gn}}}(\mathbf{x}^{[gn]}) \\ \nn &=& \sum_{\alpha \in \mathcal{I}^{\widetilde{\mathcal{F}}}} p^{[0]}_{\alpha_{\vert_{0}}}(\mathbf{x}^{[0]}) \cdots p^{[n]}_{\alpha_{\vert_{n}}}(\mathbf{x}^{[n]}) = p
\ee
where we have used Definition \ref{def:omegaG-dec} \ref{def:omegaG-dec:i} \ref{im:b} in the third equality, 
and the fact that $\alpha \mapsto {}^{g^{}} \alpha$ is a bijection on $\mathcal{I}^{\widetilde{\mathcal{F}}}$ 
and that $i \mapsto gi$ is a bijection on $[n]$ in the fifth equality.

In the converse direction, the following holds: 
If a polynomial is $G$-invariant, then it has an $(\Omega,G)$-decomposition if $G$ acts freely on $\Omega$. Moreover, every $\Omega$ can be refined so that $G$ acts freely on it (Proposition \ref{pro:extend}). 

\item 
The existence of an $(\Omega,G)$-decomposition might imply an even stronger symmetry than $G$-invariance. As we will see in Example \ref{ex:omegaG-dec} \ref{ex:omegaG-dec:i}, the existence of a $(\Sigma_n,G)$-decomposition for any transitive group action of some group $G$  already implies $S_{n+1}$-invariance. This is closely related to the action not being free.\demo
\end{enumerate}
\end{remark}

Let us now revisit our running examples---the simple and double edge of Example \ref{ex:wscbasic}---in the light of invariant decompositions. 

\begin{example}[The simple and double edge with invariance]\label{ex:omegaG-dec0}\quad\nobreakpar
\begin{enumerate}[label=(\roman*),ref=(\roman*),leftmargin=*]
\item 
On the simple edge $\Lambda_1,$ the elements in $\mathcal{I}^{\widetilde{\mathcal{F}}}$ are just single values, 
and thus the corresponding decomposition is given by
$$
p = \sum_{\alpha=1}^r p^{[0]}_{\alpha}(\mathbf{x}^{[0]}) \cdot p^{[1]}_{\alpha}(\mathbf{x}^{[1]}). 
$$
The $C_2$-invariant decomposition is given by
$$p= \sum_{\alpha=1}^r p_{\alpha}(\mathbf{x}^{[0]}) \cdot p_{\alpha}(\mathbf{x}^{[1]}).$$
\item \label{ex:omegaG-dec0:ii}
For the double edge $\Delta$ we have two facets and thus the $\Delta$-decomposition reads   
$$
p= \sum_{\alpha,\beta=1}^r p_{\alpha, \beta}^{[0]}(\mathbf{x}^{[0]}) \cdot p_{\beta,\alpha}^{[1]}(\mathbf{x}^{[1]}).
$$ 
Note that the order of the indiced $\alpha,\beta$ does not matter here, since there is no connection between the local polynomials at site $0$ and $1$.
But for the non-trival $C_2$ action,  Definition \ref{def:omegaG-dec} \ref{def:omegaG-dec:i} \ref{im:b}  specifies that 
$$
p^{[0]}_{\alpha,\beta}=p^{[1]}_{\alpha,\beta}, 
$$
so an  $(\Delta,C_2)$-decomposition is of the form
\be  
p= \sum_{\alpha,\beta=1}^r p_{\alpha, \beta}(\mathbf{x}^{[0]}) \cdot p_{\beta,\alpha}(\mathbf{x}^{[1]}).\label{eq:invdoubleedge}
\ee\demo
\end{enumerate}
\end{example}

Let us now consider an invariant polynomial on the double edge which we will revisit  in Example \ref{ex:inv-sos} in the light of sum-of-squares invariant decompositions. 

\begin{example}[Invariant polynomial on the double edge]\label{ex:omegaG-exp}
Consider the polynomial 
$$
p=x^2+y^2+4(1+xy)^2=4+8xy+x^2+y^2+4x^2y^2\in \mathbb{R}[x]\otimes\mathbb{R}[y] 
$$ 
which is  invariant with respect to the permutation of $x$ and $y$.
A  $(\Delta, C_2)$-decomposition of $p$ has the form 
$$
p=\sum_{\alpha,\beta=1}^2 p_{\alpha,\beta}(x)p_{\beta,\alpha}(y),
$$   
with
$$
p_{1,1}(t) = \frac12 +2t^2, \quad p_{1,2}(t) =p_{2,1}(t)= \sqrt{\frac{15}{8}}, \quad  p_{2,2}(t)=\sqrt{8}t. 
$$   
It is easy to see that a decomposition of rank $1$ does not exist, showing that the $(\Delta,C_2)$-rank is indeed 2. 
\demo\end{example}

Let us now see more standard examples of $(\Omega,G)$-decompositions based off the weighted simplicial complexes of Example \ref{ex:wsc}. 

\begin{example}[The simplex and the circle with their symmetry]\label{ex:omegaG-dec}\quad\nobreakpar
\begin{enumerate}[label=(\roman*),ref=(\roman*),leftmargin=*]
\item \label{ex:omegaG-dec:i}
For $n \geq 1$ consider an $n$-simplex $\Sigma_n$, whose facets are given by $\mathcal{\widetilde{F}} = \{[n] \}$.
Since $\widetilde{\mathcal{F}}$ only contains one facet encompassing all vertices, the corresponding $\Sigma_n$-decomposition is given by
$$
p= \sum_{\alpha=1}^{r} p_{\alpha}^{[0]}(\mathbf{x}^{[0]}) \cdot p_{\alpha}^{[1]}(\mathbf{x}^{[1]}) \cdots p_{\alpha}^{[n]}(\mathbf{x}^{[n]}).
$$
The minimal integer $r$ among all such decompositions is the $\rank_{\Sigma_n}(p)$---this is usually called the \emph{tensor rank}. 

Now assume there is  a group action $G$ on $[n]$ which is transitive, i.e.\ it generates  only  one orbit, namely  $Gi = [n]$ for all $i \in [n]$. 
Then  Definition \ref{def:omegaG-dec} \ref{def:omegaG-dec:i} \ref{im:b} requires
$ p^{[i]}_{\alpha} = p^{[j]}_{\alpha} $ 
for all $i,j,\alpha$, 
and hence the corresponding $(\Sigma_n, G)$-decomposition reads
$$p = \sum_{\alpha=1}^{r} p_{\alpha}(\mathbf{x}^{[0]}) \cdot p_{\alpha}(\mathbf{x}^{[1]}) \cdots p_{\alpha}(\mathbf{x}^{[n]}).$$
This decomposition is manifestly fully symmetric with respect to every permutation of $\mathbf{x}^{[i]}$ with $\mathbf{x}^{[j]}$. The minimal such $r$ is the $\rank_{(\Sigma_n,G)}(p)$---usually called the \emph{symmetric tensor rank}. 

\item \label{ex:omegaG-dec:ii}
For $n \geq 3$ consider the circle $\Theta_n$. The $\Theta_n$-decomposition of $p$ reads
$$
p = \sum_{\alpha_0, \ldots, \alpha_n=1}^{r} p_{\alpha_0, \alpha_1}^{[0]}(\mathbf{x}^{[0]}) \cdot p_{\alpha_1, \alpha_2}^{[1]}(\mathbf{x}^{[1]}) \cdots p_{\alpha_{n}, \alpha_0}^{[n]}(\mathbf{x}^{[n]}).
$$
The minimal such $r$ is the $\rank_{\Theta_n}(p)$---this is usually called the \emph{operator Schmidt rank}. 
 
Since the cyclic group $C_{n}$ acts freely on $\Theta_n$, we obtain the $(\Theta_n,C_{n})$-decomposition
$$
 p = \sum_{\alpha_0, \ldots, \alpha_n=0}^{r} p_{\alpha_0, \alpha_1}(\mathbf{x}^{[0]}) \cdot p_{\alpha_1, \alpha_2}(\mathbf{x}^{[1]}) \cdots p_{\alpha_{n}, \alpha_0}(\mathbf{x}^{[n]}).
$$
This decomposition is manifestly translational invariant, that is, invariant with respect to permutations $\mathbf{x}^{[i]} \mapsto \mathbf{x}^{[a+i]}$ for $a \in \mathbb{N}$ where the addition is   modulo $n+1$. 
Note that polynomials with such a decomposition are generally not $S_{n}$-invariant.
The minimal such $r$ is called the $\rank_{(\Theta_n,C_n)}(p)$---usually called the  \emph{translationally invariant operator Schmidt rank}.
\demo
\end{enumerate}
\end{example}

\begin{center}
\deco \:\:  \emph{Decompositions without invariance}\:\:\: \deco
\label{sssec:withoutinv}
\end{center}

The first result on the existence of polynomial decompositions does not involve any invariance. 
It is an adaption of the result for tensor decompositions (see \cite[Theorem 11]{De19d}), which we will prove here for completeness.

\begin{theorem}[Existence of $\Omega$-decompositions]
\label{thm:omega-dec}
For every connected weighted simplicial complex $\Omega$ and every $p \in \mathcal{P}$ there exists an $\Omega$-decomposition of $p$, i.e.\  $\rank_{\Omega}(p) < \infty$. 
Moreover, given  a decomposition of the form
\be \label{eq:elementaryPoly}
p = \sum_{j = 1}^{r} p_j^{[0]}(\mathbf{x}^{[0]}) \cdots p_j^{[n]}(\mathbf{x}^{[n]})
\ee
where $p_j^{[i]} \in \mathbb{R}[\mathbf{x}^{[i]}]$,  there exists an  $\Omega$-decomposition of $p$ only using the $p_j^{[i]}$ as local polynomials at each site $i$.
\end{theorem}

Note that the $\Omega$-decomposition obtained by ``reusing'' the polynomials of \eqref{eq:elementaryPoly} may not be optimal, i.e.\ it may need more terms than its rank.  

\begin{proof}
We start with an elementary polynomial decomposition
$$p = \sum_{j \in \mathcal{I}} p_j^{[0]}(\mathbf{x}^{[0]}) \cdot p_j^{[1]}(\mathbf{x}^{[1]}) \cdots p_j^{[n]}(\mathbf{x}^{[n]})$$
where $\mathcal{I}$ is a finite index set and $p_j^{[i]} \in \mathbb{R}[\mathbf{x}^{[i]}]$ for all $j \in \mathcal{I}$. For $i \in [n]$ and $\beta \in \mathcal{I}^{\widetilde{\mathcal{F}}_i}$ we define
\be \nn p_{\beta}^{[i]} \coloneqq \left\{\begin{array}{ll} p_j^{[i]} & : \beta \textrm{ takes the constant value } j \in \mathcal{I}\\[0.1cm] 0 & : \textrm{else.} \end{array}\right. \ee
Since $\Omega$ is connected, for $\alpha \in \mathcal{I}^{\mathcal{\widetilde{F}}}$ the restricted functions $\alpha_{\vert_i}$ are all constant if and only if $\alpha$ is constant. It follows that
\be  \nn \sum_{\alpha \in \mathcal{I}^{\widetilde{\mathcal{F}}}} p^{[0]}_{\alpha_{\vert_0}}(\mathbf{x}^{[0]}) \cdots p^{[n]}_{\alpha_{\vert_n}}(\mathbf{x}^{[n]}) = \sum_{j \in \mathcal{I}} p_j^{[0]}(\mathbf{x}^{[0]}) \cdots p_j^{[n]}(\mathbf{x}^{[n]}) =p(\mathbf{x}^{[0]}, \ldots, \mathbf{x}^{[n]})  \ee
 is an $\Omega$-decomposition of $p$.
\end{proof}

\begin{center}
\deco \:\:  \emph{Invariant decompositions with free group actions}\:\:\: \deco
\label{sssec:free}
\end{center}

We now show that if $G$ acts freely on $\Omega$, then every $G$-invariant polynomial admits an $(\Omega,G)$-decomposition. Recall that `free' was defined in Definition \ref{def:action} \ref{im:action:iii}.
The proof   is similar to that of \cite[Theorem 13]{De19d}, but we include it here for completeness. 
We will illustrate the idea of the proof in Example \ref{ex:omegaG-constr}.

\begin{theorem}[Invariant decompositions with free group actions]
\label{thm:omegaG-dec}
Let $\Omega$ be a connected weighted simplicial complex, $G$ a group action on $\Omega$, and  $p \in \mathcal{P}$ a $G$-invariant polynomial. 
If $G$ acts freely on $\Omega$, then $p$ has an $(\Omega,G)$-decomposition, i.e.\ $\rank_{(\Omega,G)}(p) < \infty$.
Moreover, given a decomposition of the form \eqref{eq:elementaryPoly}, 
an $(\Omega,G)$-decomposition of $p$ can be obtained by using only nonnegative multiples of the $p_j^{[i]}$ as local polynomials at each site $i$.
\end{theorem}

As in Theorem \ref{thm:omega-dec}, the $(\Omega,G)$-decomposition obtained by  ``reusing'' the polynomials of \eqref{eq:elementaryPoly} will generally not be optimal. 

Note that  every weighted simplicial complex $\Omega$ can be refined so that $G$ acts freely (by Proposition \ref{pro:extend}),  and refining will  translate to adding more summation indices in the polynomial decomposition, as  in Example \ref{ex:omegaG-dec0} \ref{ex:omegaG-dec0:ii}. 

The idea of the proof is simple.  
Starting from the decomposition in \eqref{eq:elementaryPoly}, we essentially build 
$$
 \frac{1}{\vert G\vert} \sum_{g\in G}gp = p 
$$
where $gp$ is defined in \eqref{eq:gp}, and   let  $g$ act on each of the local terms in the decomposition. 
The latter    can then be transformed into an $(\Omega,G)$-decomposition of $p$.  

\begin{proof}
Since $G$ acts freely, by Remark \ref{rem:grac} \ref{rem:grac:iv}, there exists a $G$-linear map $\textbf{z}: \widetilde{\mathcal{F}} \to G$, where $G$ acts on itself by left-multiplication. In the following, we fix one such mapping. For the polynomial $p$ we first obtain by Theorem \ref{thm:omega-dec} an $\Omega$-decomposition and denote the local elements by
\be \nn 
Q^{[i]} \coloneqq \left(q_\beta^{[i]}(\mathbf{x}^{[i]}) \right)_{\beta \in \mathcal{I}^{\widetilde{\mathcal{F}}_i}}  
\ee
where $q_\beta^{[i]}(\mathbf{x}^{[i]}) \in \mathbb{R}[\mathbf{x}^{[i]}]$ for every $i \in [n]$. We define a new index set
$$
\hat{\mathcal{I}} \coloneqq \mathcal{I} \times G
$$
together with the projection maps $\pi_1: \hat{\mathcal{I}} \to \mathcal{I}$ and $\pi_2: \hat{\mathcal{I}} \to G$. For each $i \in [n]$ and $\beta \in \hat{\mathcal{I}}^{\mathcal{\widetilde{F}}_i}$ we now define the following local polynomials:
$$
p_{\beta}^{[i]} \coloneqq \left\{ \begin{array}{ll} q^{[gi]}_{{}^g (\pi_1 \circ \beta)}(\mathbf{x}^{[i]}) & : \pi_2 \circ \beta = ({}^{g^{-1}} \mathbf{z})_{\vert_i} \\0 &: \textrm{else.} \end{array} \right.
$$
Note that $p_{\beta}^{[i]}(\mathbf{x}^{[i]})$ is well-defined since $g$ is uniquely determined by the relation $\pi_2 \circ \beta = ({}^{g^{-1}} \mathbf{z})_{\vert_i}$ if such a $g$ exists. This is due to the fact that if $({}^{g_1^{-1}} \mathbf{z})_{\vert_i} = ({}^{g_2^{-1}} \mathbf{z})_{\vert_i}$ we have $g_1 \cdot \mathbf{z}(F) = g_2 \cdot \mathbf{z}(F)$ for any $F \in \widetilde{\mathcal{F}}_i$ by $G$-linearity of $\mathbf{z}$. 
But this implies that $g_1 = g_2$.
In addition, the defined local polynomials fulfil Definition \ref{def:omegaG-dec} \ref{def:omegaG-dec:i} \ref{im:b} since for $g,h \in G$ we obtain
$$
p^{[hi]}_{{}^{h} \beta}(\mathbf{x}^{[i]}) = q^{[ghi]}_{{}^g (\pi_1 \circ {}^{h} \beta)}(\mathbf{x}^{[i]}) = q^{[ghi]}_{{}^{gh} (\pi_1 \circ \beta)}(\mathbf{x}^{[i]}) = p^{[i]}_{\beta}(\mathbf{x}^{[i]})
$$
using the fact that $\pi_2\circ {}^h\beta=\left({}^{g^{-1}}\mathbf z\right)_{\mid_{hi}}$ is equivalent to $\pi_2\circ\beta=\left({}^{(gh)^{-1}}\mathbf z\right)_{\mid_i}.$ 

It only remains to show that the local polynomials form an $(\Omega,G)$-decomposition of $p$. To this end we compute
\be\nn 
\sum_{\hat{\alpha} \in \hat{\mathcal{I}}^{\widetilde{\mathcal{F}}}} p_{\hat{\alpha}_{\vert_0}}^{[0]}(\mathbf{x}^{[0]}) \cdots p_{\hat{\alpha}_{\vert_n}}^{[n]}(\mathbf{x}^{[n]}) &=& \sum_{\substack{z \in G^{\widetilde{\mathcal{F}}} \\ \forall i \exists g_i: z_{\vert_i} = \left({}^{g_i^{-1}} \mathbf{z}\right)_{\vert_i}}} \sum_{\alpha \in \mathcal{I}^{\widetilde{\mathcal{F}}}} q^{[g_0 0]}_{{}^{g_0} (\alpha_{\vert_0})}(\mathbf{x}^{[0]}) \cdots q^{[g_n n]}_{{}^{g_n} (\alpha_{\vert_n})}(\mathbf{x}^{[n]}).
\ee 
Using that $\Omega$ is connected and $\mathbf{z}$ is $G$-linear, for each $z$ fulfilling the conditions from the outer sum on the right, we obtain $g_i = g_j \eqqcolon g$ for all $i,j \in [n]$. So the corresponding inner sum becomes 
\be\nn   
\sum_{\alpha \in \mathcal{I}^{\widetilde{\mathcal{F}}}} q^{[g 0]}_{{}^{g} (\alpha_{\vert_0})}(\mathbf{x}^{[0]}) \cdots q^{[g n]}_{{}^{g} (\alpha_{\vert_n})}(\mathbf{x}^{[n]})= p(\mathbf{x}^{[g^{-1}0]}, \ldots, \mathbf{x}^{[g^{-1}n]})= p(\mathbf{x}^{[0]}, \ldots, \mathbf{x}^{[n]}), 
\ee 
using $G$-invariance of $p$. 
Hence the total sum equals a positive multiple of  $p$, where the factor is the number of all $z$ which fulfill the above conditions. In fact,  this number is just $\vert G\vert$, since the ${}^{g^{-1}}\mathbf z$ for $g\in G$ are precisely the different possible choices for $z$.
So  dividing by $\vert G\vert$  and absorbing its positive $(n+1)$-th root into the local polynomials yields an $(\Omega,G)$-decomposition of $p$. The last statement is immediate by construction.
\end{proof}

The following are some immediate and useful relations between ranks: 

\begin{corollary}[Relations among  ranks]
\label{cor:ranks}
Let $\Omega$ be connected and $G$  a free group action on $\Omega$, 
and $\Sigma_n$ the simplex (defined in Example \ref{ex:wsc} \ref{ex:wsc:i}). 
Then for every $G$-invariant $p \in \mathcal{P}$  we have
$$ 
\rank_{(\Omega,G)}(p) \leq \vert G\vert \cdot \rank_{\Omega}(p) \leq \vert G\vert \cdot \rank_{\Sigma_n}(p).
$$
\end{corollary}

In words, the first inequality says that one can impose invariance 
by increasing the rank by at most $\vert G\vert$, i.e.\ imposing invariance ``costs'' at most $\vert G\vert$ 
(as long as $G$ is free, else one cannot impose invariance within our framework). 
The second inequality says that the tensor rank is always the most expensive rank, 
i.e.\ having one joint index is the most costly decomposition.

\begin{proof}
The first inequality is immediate from the construction in the  proof of Theorem \ref{thm:omegaG-dec}, and the second inequality follows from the construction in the proof of Theorem \ref{thm:omega-dec}.
\end{proof}

Let us now illustrate the proof of Theorem \ref{thm:omegaG-dec} for the double edge. 

\begin{example}[Invariant decomposition on the double edge]\label{ex:omegaG-constr}
The cyclic group $C_2$ provides  a free group action on the double edge $\Delta$, 
so every $C_2$-invariant polynomial admits a $(\Delta, C_2)$-decomposition, given by Equation \eqref{eq:invdoubleedge}. Let us now construct it. 

For the group action of $C_2 = \{e, c\}$  on $\widetilde{\mathcal{F}} = \{\mathfrak{a}, \mathfrak{b}\}$ (with $c \mathfrak{a} = \mathfrak{b}$) there exists a $G$-linear map 
$\mathbf{z}: \widetilde{\mathcal{F}} \to G$, which can be chosen as
$$ 
\begin{array}{r l} \mathbf{z}:& \mathfrak{a} \mapsto e \\ & \mathfrak{b} \mapsto c.\end{array}
$$
(There is exactly one other choice, namely   exchanging the two outcomes of $\mathbf{z}$.) 

We start with a $\Delta$-decomposition of $p$, namely 
$$
p = \sum_{\alpha,\beta = 1}^{r} q^{[0]}_{\alpha, \beta}(\mathbf{x}^{[0]}) \cdot q^{[1]}_{\beta, \alpha}(\mathbf{x}^{[1]}), 
$$
where we associate the index $\alpha$ with $\mathfrak{a}$ and $\beta$ with $\mathfrak{b}$. To construct a $(\Delta, C_2)$-decomposition, we extend the indices $\alpha, \beta$ to tuples $(\alpha, g_0)$, $(\beta, g_1)$ where $g_0, g_1 \in C_2$. We define the local polynomials as
$$ 
p^{[0]}_{(\alpha, g_0), (\beta, g_1)}(\mathbf{x}^{[0]}) \coloneqq \left\{\begin{array}{ll}q^{[0]}_{\alpha,\beta}(\mathbf{x}^{[0]}) & \textrm{ if } (g_0, g_1) = (e,c)\\[0.2cm] q^{[1]}_{\beta,\alpha}(\mathbf{x}^{[0]}) & \textrm{ if } (g_0, g_1) = (c,e)\\[0.2cm] 0 & \textrm{ else}  \end{array}\right.
$$
and
$$ 
p^{[1]}_{(\alpha,g_0),(\beta, g_1)}(\mathbf{x}^{[1]}) \coloneqq \left\{\begin{array}{ll}q^{[1]}_{\alpha,\beta}(\mathbf{x}^{[1]}) & \textrm{ if } (g_0, g_1) = (e,c)\\[0.2cm] q^{[0]}_{\beta,\alpha}(\mathbf{x}^{[1]}) & \textrm{ if } (g_0, g_1) = (c,e)\\[0.2cm] 0 & \textrm{ else.}  \end{array}\right.
$$
 For $\alpha, \beta \in \{1, \ldots, r\}$ and $g_0, g_1 \in C_2$, the symmetry condition gives rise to the definition
$$
p^{[c0]}_{{}^c((\alpha,g_0), (\beta,g_1))} = p^{[1]}_{(\beta,g_1),(\alpha,g_0)} = p^{[0]}_{(\alpha,g_0), (\beta,g_1)} \eqqcolon p_{(\alpha,g_0), (\beta,g_1)}.
$$
In addition, it is easy to verify that
\be \nn & & \sum_{g_0, g_1 \in C_2} \sum_{\alpha,\beta = 1}^{r} p_{(\alpha,g_0), (\beta,g_1)}(\mathbf{x}^{[0]}) \cdot p_{(\beta,g_1), (\alpha,g_0)}(\mathbf{x}^{[1]}) \\ &=& p(\mathbf{x}^{[0]}, \mathbf{x}^{[1]}) + p(\mathbf{x}^{[1]}, \mathbf{x}^{[0]}) = 2 p(\mathbf{x}^{[0]}, \mathbf{x}^{[1]})\ee
which shows that the local polynomials $\frac{1}{\sqrt{2} }\cdot p_{(\alpha,g_0), (\beta,g_1)}$ form  a $(\Delta, C_2)$-decomposition of $p$. This also  implies  $ {\rm rank}_{(\Delta,C_s)}(p)\leq 2 \cdot r$. \demo
\end{example}

\bigskip

\begin{center}
\deco \:\:  \emph{Invariant decompositions with blending group actions}\:\:\: \deco 
\label{sssec:blending}
\end{center}
Since the full symmetry group $S_{n+1}$ is not free on the simplex $\Sigma_n$, Theorem \ref{thm:omegaG-dec} does not say anything about the existence of $(\Sigma_n, S_{n+1})$-decompositions. 
In fact, for real polynomials, such decompositions may not exist (see Example \ref{exa:extrasign}). 
Nonetheless, we can prove another, weaker existence result for polynomial decompositions with a blending group action $G$ (Theorem \ref{thm:omegaG-diff}), where `blending' was defined in Definition \ref{def:action} \ref{im:action:iv}. 
In preparation for this result we need the following two lemmas.
The first lemma introduces a ``negative part'' in the symmetric decomposition, which can be omitted if $n$ is even:

\begin{lemma}[Symmetric decompositions for tensors \cite{Co08c}]
\label{lem:symtensor}
Let $T \in \mathbb{R}^d \otimes \cdots \otimes \mathbb{R}^d \cong \mathbb{R}^{(n+1)d}$ be $S_{n+1}$-invariant, i.e.\ for every $i_0, \ldots, i_n \in \{1,\ldots,d\}$ and permutation $\sigma \in S_{n+1}$ we have
$$
T_{i_0, \ldots, i_n} = T_{\sigma(i_0), \ldots, \sigma(i_n)}.
$$ 
Then there exist $r_1, r_2 \in \mathbb{N}$ and $v_1, \ldots, v_{r_1}, v_{r_1+1}, \ldots, v_{r_1+r_2} \in \mathbb{R}^d$ such that
\be \label{eq:minus}
T = \sum_{\ell=1}^{r_1} v_\ell^{\otimes n+1} - \sum_{\ell=r_1+1}^{r_1+r_2} v_\ell^{\otimes n+1}
\ee
If $n$ is even, there  exists a decomposition
$$T = \sum_{\ell=1}^{r_1} v_\ell^{\otimes n+1}.
$$
\end{lemma}

The last statement is not given in \cite{Co08c}, but it is obvious, since the minus sign can be absorbed into the odd number of terms $n+1$ (because $(-1)^{n+1} = -1$). 

The minus sign in Equation \eqref{eq:minus} is necessary, for consider the simple case of real matrices, namely when the tensor $T$ lives in the space $\mathbb{R}^d \otimes \mathbb{R}^d \cong \mathcal{M}_d(\mathbb R)$. 
Without a minus sign,  Equation \eqref{eq:minus}  would read 
$$ 
T = \sum_{\ell=1}^{r_1} v_\ell \otimes v_\ell = \sum_{\ell=1}^{r_1} v_\ell v_\ell^{t} \geqslant 0
$$ 
(where we have used that $v \otimes w = v w^t$), implying that every symmetric matrix is positive semidefinite. This is false, so the minus sign is crucial. 
The importance of the minus sign will be illustrated in Example \ref{exa:extrasign}.

The second lemma states the subadditivity and submultiplicativity of the $(\Omega,G)$-rank, and is proven in \cite[Proposition 16]{De19d}.  

\begin{lemma}[Subadditivity and submultiplicativity of ranks \cite{De19d}]
\label{lem:sum}
Let $p_1, p_2 \in \mathcal{P}$.  
\begin{enumerate}[label=(\roman*),leftmargin=*]
	\item $\rank_{(\Omega,G)}(p_1 + p_2) \leq \rank_{(\Omega,G)}(p_1) + \rank_{(\Omega,G)}(p_2)$
	\item $\rank_{(\Omega,G)}(p_1 \cdot p_2) \leq \rank_{(\Omega,G)}(p_1) \cdot \rank_{(\Omega,G)}(p_2)$
\end{enumerate}
\end{lemma}

We are now ready for the existence of  invariant decompositions with blending group actions. 

\begin{theorem}[Invariant decompositions with blending group actions]
\label{thm:omegaG-diff}
Let $\Omega$ be a connected weighted simplicial complex, and $G$ a blending group action on $\Omega$. 
For any $G$-invariant $p \in \mathcal{P}$  there exist two polynomials $q_1, q_2 \in \mathcal{P}$ with $p = q_1 - q_2$, where both have an $(\Omega,G)$-decomposition. If $n$ is even we can  set $q_2 = 0$.
\end{theorem}

\begin{proof}
We start with a non-invariant decomposition of $p$, as given in Equation \eqref{eq:elementaryPoly}, 
where $\mathcal{I}$ is a finite index set. Now we choose real numbers $d_\ell^{[i]} \in \mathbb{R}$ for $i \in [n]$ and $\ell \in \{1, \ldots, r_1 + r_2\}$, such that the following holds:
$$\sum_{\ell=1}^{r_1} d_\ell^{[i_0]} \cdots d_\ell^{[i_n]} - \sum_{\ell=r_1+1}^{r_2} d_\ell^{[i_0]} \cdots d_\ell^{[i_n]} = \left\{\begin{array}{ll}1 & : \{i_0, \ldots, i_n\} = [n]\\0 & :\textrm{else} \end{array} \right.$$
This is possible since the tensor on the right hand side is real and symmetric, hence the existence follows by Lemma \ref{lem:symtensor}.
For $i \in [n]$, $\ell \in \{1, \ldots, r_1+r_2\}$ and $\beta \in \mathcal{I}^{\widetilde{\mathcal{F}}_i}$ we define
\be \nn p_{\ell, \beta}^{[i]}(\mathbf{x}^{[i]}) \coloneqq \left\{\begin{array}{ll} \sum_{g \in G} d_{\ell}^{[gi]} p_{j}^{[gi]}(\mathbf{x}^{[i]}) &: \beta \textrm{ takes the  constant value  } j \in \mathcal{I}\\[0.2cm] 0 & \textrm{: else} \end{array} \right. .\ee
For fixed $\ell$, the polynomials $p_{\ell, \beta}^{[i]}$ fulfil Definition \ref{def:omegaG-dec} \ref{def:omegaG-dec:i} \ref{im:b} and hence give rise to $(\Omega,G)$-decompositions of polynomials $p_1, \ldots, p_{r_1}, p_{r_1+1}, \ldots, p_{r_1+r_2}$.

We now define $q_1$ as 
\be \nn q_1 &\coloneqq& \sum_{\ell=1}^{r_1} p_{\ell}\ =\ \sum_{\ell=1}^{r_1} \sum_{\alpha \in \mathcal{I}^{\widetilde{\mathcal{F}}}} p_{\ell, \alpha_{\vert_0}}^{[0]}(\mathbf{x}^{[0]}) \cdots p_{\ell, \alpha_{\vert_n}}^{[n]}(\mathbf{x}^{[n]}) \\ \nn &=& \sum_{g_0,\ldots, g_n \in G} \sum_{\ell = 1}^{r_1} d_{\ell}^{[g_0 0]} \cdots d_{\ell}^{[g_n n]} \sum_{j \in \mathcal{I}} p_j^{[g_0 0]}(\mathbf{x}^{[0]}) \cdots p_j^{[g_n n]}(\mathbf{x}^{[n]}) 
 \ee
where we have used that $\Omega$ is connected for the third equality, and thus $\alpha_{\vert_i}$ constant for all $i$ if and only if $\alpha$ is constant. 
Note that $q_1$ has an $(\Omega,G)$-decomposition by Lemma \ref{lem:sum}, since all $p_\ell$ do.
We define $q_2$ similarly as
\be \nn q_2 \coloneqq \sum_{\ell=r_1+1}^{r_2} p_{\ell}. 
\ee

Because of the definition of $d_{\ell}^{[i]}$, and the fact that the action of  $G$ is blending, the difference $q_1 - q_2$ simplifies to
\be \nn q_1 - q_2&=& \sum_{\substack{g_0, \ldots, g_n \in G \\ \{g_0 0, \ldots, g_n n\} = [n]}} \sum_{j \in \mathcal{I}} p_j^{[g_0 0]}(\mathbf{x}^{[0]}) \cdots p_j^{[g_n n]}(\mathbf{x}^{[n]}) \\ \nn &\sim& \sum_{g \in G} \sum_{j \in \mathcal{I}} p_j^{[g 0]}(\mathbf{x}^{[0]}) \cdots p_j^{[g n]}(\mathbf{x}^{[n]}) \\ \nn &=& \vert G\vert \cdot p \ee
where $\sim$ stands for positive multiple of. 
Note that we have used  that $p$ is $G$-invariant in the last equality. 
Dividing by $\vert G\vert$ and the positive scaling factor proves the statement, since the scaling can be absorbed in the local polynomials.

The last statement of the Theorem follows from the statement in Lemma \ref{lem:symtensor} for even $n$.
\end{proof}

\begin{example}[The minus sign in the single and double edge]\label{exa:extrasign}
The  minus sign in the decomposition of Theorem \ref{thm:omegaG-diff} is   necessary (as long as we do not switch to complex coefficients). For example, the polynomial $ p = x^2+y^2$ is $C_2$-invariant, and since $C_2$ is blending on the single edge $\Lambda_1$, 
there exists an $(\Lambda_1, C_2)$-decomposition for $p$  with this additional minus sign (by Theorem \ref{thm:omegaG-diff}): 
$$ 
p= x^2+y^2 = p_1(x) \cdot p_1(y) - p_2(x) \cdot p_2(y)
$$
where
$$
p_1(t) = \frac{1}{\sqrt{2}}(1+t^2)\ \mbox{ and } \ p_2(t) = \frac{1}{\sqrt{2}}(1-t^2).
$$
But  for degree reasons there cannot exist an actual $(\Lambda_1, C_2)$-decomposition for $p$, i.e.\ an invariant decomposition without the additional minus sign.

On the other hand, the refinement of $\Lambda_1$ to the double edge $\Delta$ allows for  a free group action of $C_2$. Hence  there exists a $(\Delta,C_2)$-decomposition of $p$ (by Theorem \ref{thm:omegaG-dec}), given for example by
$$ 
x^2+y^2 = \sum_{\alpha, \beta = 1}^{2} p_{\alpha, \beta}\left(x\right) \cdot p_{\beta, \alpha}\left(y\right)$$
where $p_{1,1}(t) = 0, p_{1,2}(t) = t^2, p_{2,1}(t) = 1$ and $p_{2,2}(t) = 0.$ This shows that ${\rm rank}_{(\Delta,C_2)}(p)=2.$\demo
\end{example}

\begin{example}[Fully symmetric polynomials]
Since the action of the full permutation group is blending, every fully symmetric polynomial $p \in \mathbb{R}[x_0, x_1, \ldots, x_n]$ can be written as a difference of two polynomials with $(\Sigma_n,S_{n+1})$-decompositions, i.e.\ 
$$p= \sum_{\ell=1}^{r_1} p_{\ell}(x_0) \cdots p_{\ell}(x_n) - \sum_{\ell=r_1 + 1}^{r_1 + r_2} p_{\ell}(x_0) \cdots p_{\ell}(x_n).
$$
Since $p_{\ell}$ is a univariate polynomial, it is given by a vector of coefficients $(c_{\ell,k})_{k=1}^{d}$, namely
$$p_{\ell}(t) = \sum_{k=0}^{d} c_{\ell,k} t^k.$$
This leads to a decomposition into monomial symmetric polynomials $\mathfrak{m}_{\alpha}(x_0, \dots, x_n)$ with $\alpha \in \mathbb{N}^{n+1}$, which is defined as the sum over all monomials $x_0^{\beta_0} \cdot x_1^{\beta_1} \cdots x_n^{\beta_n}$ where $\beta$ ranges over all distinct permutations of $(\alpha_0, \ldots, \alpha_n)$.
Spelling out the  $(\Sigma_n, S_{n+1})$-decompositions we obtain the following the decomposition into monomial symmetric polynomials: 
$$ p = \sum_{0 \leq \alpha_0 \leq \alpha_1 \ldots \leq \alpha_n \leq d}  \left(\sum_{\ell=1}^{r_1} c_{\ell,\alpha_0}\cdots c_{\ell,\alpha_n}   - \sum_{\ell=r_1+1}^{r_1+r_2}   c_{\ell,\alpha_0}\cdots c_{\ell,\alpha_n}   \right)\mathfrak{m}_{(\alpha_0, \ldots, \alpha_n)}.$$

Conversely, given $p$ as a linear combination of monomial symmetric polynomials
$$p = \sum_{0 \leq \alpha_0 \leq \alpha_1 \leq \ldots, \leq \alpha_n \leq d} D_{\alpha_0, \ldots, \alpha_n} \mathfrak{m}_{(\alpha_0, \ldots, \alpha_n)}$$
we obtain the $(\Sigma_n, S_{n+1})$-decompositions by means of a symmetric tensor decomposition of the symmetrically completed tensor $D$.
\demo
\end{example}

\subsection{The invariant separable decomposition}
\label{ssec:sepdec}

In this section we assume that every local space of polynomials is equipped with a convex cone $\mathcal{C}^{[i]} \subseteq \mathbb{R}[\mathbf{x}^{[i]}]$, i.e.\ a set which fulfills $\alpha p + \beta q \in \mathcal{C}$ for all $p,q \in \mathcal{C}$ and $\alpha, \beta \geq 0$.
Important examples of such cones are the cone of sum-of-squares (sos) polynomials
$$ \mathcal{C}_{\mathrm{sos}} \coloneqq \left\{p \in \mathbb{R}[\mathbf{x}]: p = \sum_{k=1}^N q_k^2 \textrm{ for some } q_k \in \mathbb{R}[\mathbf{x}], N \in \mathbb{N}\right\},$$
the cone of nonnegative polynomials
$$\mathcal{C}_{\mathrm{nn}} \coloneqq \left\{p \in \mathbb{R}[\mathbf{x}]: p(a) \geq 0 \textrm{ for all } a \in \mathbb{R}^{m}\right\},$$
and the cone of polynomials with nonnegative coefficients
$$\mathcal{C}_{\mathrm{nn-coeff}} \coloneqq \left\{p \in \mathbb{R}[\mathbf{x}]: p = \sum_{\alpha_1, \ldots, \alpha_m = 1}^{d} c_{\alpha_1, \ldots, \alpha_m} x_1^{\alpha_1} \cdots  x_m^{\alpha_m} \textrm{ with }  \textrm{ all } c_{\alpha_1, \ldots, \alpha_m} \geq 0 \right\}.$$
For a given set of local cones $\mathcal{C}^{[0]}, \ldots, \mathcal{C}^{[n]}$ we define the global separable cone
\be \nn \mathcal{C}_{\mathrm{sep}} &\coloneqq& \mathcal{C}^{[0]} \otimes  \mathcal{C}^{[1]} \otimes \cdots \otimes  \mathcal{C}^{[n]}\\ \nn &\coloneqq& \left\{\sum_{j=1}^{r} p_{j}^{[0]} \cdots p_{j}^{[n]}: r \in \mathbb{N}, p_{j}^{[i]} \in \mathcal{C}^{[i]} \right\} \subseteq \mathcal{P}. \ee
This is the smallest global convex cone generated by the elementary tensors formed from the local cones.
For a given group action of $G$ on $\Omega$, we further assume that $\mathcal{C}^{[i]} = \mathcal{C}^{[gi]}$ for all $g \in G$ (again we suppress the canonical isomorphism between the local polynomial spaces in the notation).

We now define and study the invariant separable decomposition of polynomials, 
i.e.\ decompositions which are inherently $G$-invariant, 
and where the containment in the separable cone is explicit---i.e.\ a   positive combination of elementary polynomials where each factor is in the local cone.

\begin{definition}[Invariant separable decomposition]\label{def:sepdec} 
Let $p \in \mathcal{C}_{\mathrm{sep}}$. 
\begin{enumerate}[label=(\roman*),ref=(\roman*),leftmargin=*]
\item 
A \emph{separable $(\Omega,G)$-decomposition} of $p$ is an $(\Omega,G)$-decomposition 
$$
\mathcal{P}^{[i]} \coloneqq \left(p_{\beta}^{[i]}\right)_{\beta \in \mathcal{I}^{\widetilde{\mathcal{F}}_i}}$$ with the additional restriction that 
$$p_{\beta}^{[i]} \in \mathcal{C}^{[i]}$$
for all $i \in [n]$ and $\beta \in \mathcal{I}^{\widetilde{\mathcal{F}}_i}$.

\item  
The minimal cardinality of $\mathcal{I}$ among all separable $(\Omega,G)$-decomposition of $p$ is called the \emph{separable $(\Omega,G)$-rank of $p$}, denoted $ \seprank_{(\Omega,G)}(p).$
If $p$ does not admit an $(\Omega,G)$-decomposition, we set $\seprank_{(\Omega,G)}(p) =  \infty$.

\item  If $G$ is the trivial group action, 
we call the separable $(\Omega,G)$-decomposition just \emph{separable $\Omega$-decomposition}, 
and its minimal number terms the \emph{separable rank}, denoted $\seprank_{\Omega}$.  
\end{enumerate}
\end{definition}

We now show the existence of invariant separable decompositions with free group actions.  
This follows from Theorem \ref{thm:omegaG-dec}, 
since it can be constructed via positive multiples of the initial decomposition.

\begin{theorem}[Invariant separable decompositions with free group actions]\label{thm:sep}
Let $\Omega$ be a connected weighted simplicial complex with a free action from the group $G$. 
Every  $G$-invariant $p \in \mathcal{C}_{\mathrm{sep}}$ admits  a separable $(\Omega,G)$-decomposition.
\end{theorem}

\begin{proof}
Let $p$ be decomposed as in Equation \eqref{eq:elementaryPoly} 
with $p_j^{[i]} \in \mathcal{C}^{[i]}$, which is a separable decomposition of $p$. 
Applying the construction of  the proof of Theorem \ref{thm:omegaG-dec} we obtain a separable $(\Omega,G)$-decomposition, since all local polynomials $p_{\beta}^{[i]}$ are positive multiples of $p_{j}^{[gi]}$ for $g \in G$. Since the local cones coincide on the orbits of $G$, this guarantees that $p_{\beta}^{[i]} \in \mathcal{C}^{[i]}$.
\end{proof}

\begin{example}[Invariant separable decomposition on the double edge]
The $(\Delta,C_2)$-decomposition of $p=x^2+y^2$ given in Example \ref{exa:extrasign} is in fact an invariant separable decomposition with respect to the local sos cones, proving that $\seprank_{(\Delta,C_2)}(p)=\rank_{(\Delta,C_2)}(p)=2.$\demo
\end{example}

We can now easily promote the results of Corollary \ref{cor:ranks} to the (invariant) separable ranks. The proof is analogous. 

\begin{corollary}[Relation between separable ranks]
\label{cor:sep-ranks}
Let $\Omega$ be connected and $G$ a free group action on $\Omega$. Then for every $G$-invariant $p \in \mathcal{P}$  we have 
$$ \seprank_{(\Omega,G)}(p) \leq \vert G\vert \cdot \seprank_{\Omega}(p) \leq \vert G\vert \cdot \seprank_{\Sigma_n}(p).$$
\end{corollary}

An analogue of Theorem \ref{thm:sep} for blending group actions is not true! 
One reason is that, if the action is blending, we cannot construct a decomposition using the local polynomials from the initial tensor  decomposition. 
This is visible already in the simplest case, namely for $(\Lambda_1, C_2)$-decompositions, as illustrated in Example \ref{exa:extrasign}. 
Another reason is that Theorem \ref{thm:omegaG-diff} (with blending group actions) uses a \emph{difference} of two $(\Omega,G)$-decompositions, and a difference of  separable elements is in general not separable.

Finally we show that the global cone of sos polynomials $\mathcal{C}_{\mathrm{sos}}$ is strictly larger than the cone of separable polynomials over local sos polynomials $\mathcal{C}_{\mathrm{sep}} =\mathcal{C}^{[0]}_{\mathrm{sos}} \otimes \cdots \otimes \mathcal{C}^{[n]}_{\mathrm{sos}}$. 
In other words, there exist polynomials which admit a sos decomposition over all variables, 
but cannot be written as tensor decomposition where every term is a sos polynomial. 
This is even true for polynomials in two variables $x$ and $y$, as the following example shows. 
The example relies on the \emph{Gram map}, which will be the cornerstone of invariant sos decompositions  (Section \ref{ssec:sosdec}). 

\begin{example}[Sos polynomials which are not separable]
\label{ex:sepPoly}
We consider the following  Gram map $\mathcal{G}$ between real-valued matrices $M \in \mathcal{M}_{2} \otimes \mathcal{M}_{2}$ and polynomials $p \in \mathbb{R}[x,y]$:  
$$\mathcal{G}:M \mapsto p \coloneqq \mathfrak{m}_1(x)^t \otimes \mathfrak{m}_1(y)^t \cdot M \cdot \mathfrak{m}_1(x) \otimes \mathfrak{m}_1(y)
$$  
where $\mathfrak{m}_1(x) \coloneqq (1, x)^t$ is the monomial basis in $x$ of degree at most $1$.

It is well-known (and easy to see) that for  $\deg_{\mathrm{loc}}(p) \leq 2$ we have $p \in \mathcal{C}_{\mathrm{sos}}$  if and only if there exists a positive semidefinite $M \in \mathcal{M}_{2} \otimes \mathcal{M}_{2}$ with  $\mathcal{G}(M) =p$.
Further, $p \in \mathcal{C}_{\mathrm{sep}}$ if and only if there exists an $M \in \mathcal{M}_{2} \otimes \mathcal{M}_{2}$ such that 
$$M = \sum_{j=1}^{r} M^{[0]}_j \otimes M^{[1]}_j$$ where all $M^{[i]}_j$ are positive semidefinite and $\mathcal{G}(M) =p$. 

For example, consider the matrix
$$ 
M = \sum_{ij=1}^{2} E_{ij} \otimes E_{ij} = b \cdot b^t = \left(\begin{array}{cccc} 1 & 0 & 0 & 1\\ 0 & 0 & 0 & 0\\ 0 & 0 & 0 &0 \\ 1 & 0 & 0 & 1 \end{array}\right)
$$
where $b= \left(e_1 \otimes e_1 + e_2 \otimes e_2 \right) \in \mathbb{R}^2 \otimes \mathbb{R}^2$ is known in the quantum information community as an (unnormalized) Bell state. 
Note that $M$ is positive semidefinite but not separable, which can easily be seen with the celebrated positive partial transposition criterion \cite{Pe96}. Furthermore, $M$ is the \emph{only} positive semidefinite matrix representing the polynomial
$$
p = 1 + 2xy + x^2 y^2 = (1+xy)^2 = \mathcal{G}(M),
$$
since  the matrix 
$$
M_{\alpha} = \left(\begin{array}{cccc} 1 & 0 & 0 & 1-\alpha \\ 0 & 0 & \alpha & 0\\ 0 & \alpha & 0 &0 \\ 1-\alpha & 0 & 0 & 1 \end{array}\right)
$$
is not positive semidefinite for any $\alpha \in \mathbb{R} \setminus \{0 \}$, 
 and $\mathcal{G}^{-1}(\{p\}) = \{M_{\alpha}: \alpha \in \mathbb{R}\}$.
This implies that $p = (1+xy)^2$ is sos but not separable with respect to the local sos cones. 

More generally, in order to show that a polynomial is sos but not separable, one needs to show that every positive semidefinite matrix $M$ with $\mathcal{G}(M) = p$ is not separable. This is generally a hard problem. 
\demo
\end{example}

\subsection{The invariant sum-of-squares decomposition}
\label{ssec:sosdec}

In this section we introduce a sum-of-squares (sos) decomposition in the $(\Omega,G)$-framework. 
To start off, notice that not every $G$-invariant sos polynomial $p$ can be decomposed into $G$-invariant polynomials $q_k$ via $p = \sum_{k=1}^{N} q_k^2,$
as the following example shows.

\begin{example}[Non-existence of stringent invariant sos decomposition]
Consider again $p = x^2+y^2$, which is obviously sos and  $C_2$-invariant, i.e.\ invariant with respect to permuting $x$ and $y$. Yet, there does \emph{not} exist a decomposition 
$$p = \sum_{k=1}^{N} q_k^2 \qquad \textrm{where all $q_k$ are $C_2$-invariant}.
$$
To see this, assume the contrary. Since $\deg(q_k) \leq \frac{1}{2} \deg(p)$, each polynomial can be written as $q_k = a_k x + a_k y + b_k$. Further, since $p$ has no constant term, we must have  $b_k = 0$. But this is impossible, since the $xy$ coefficient of $p$ is zero.
\demo\end{example}

We call the previous definition of a `stringent' invariant sos decomposition, and now introduce a more `relaxed' one, which allows for permutations among elements of the family $\{q_k\}$, and which is the correct notion as far the existence results are concerned, as we will later show.  
So let $G$ act on $[n]$, and equip the finite index set  $\mathcal{S} =  \mathcal{S}_0 \times \ldots \times \mathcal{S}_n$ with the induced group action 
$$
g \mathbf{k} \coloneqq (k_{g^{-1} 0}, \ldots, k_{g^{-1} n})
$$ 
for every $\mathbf{k} =(k_{0} \ldots, k_{n}) \in \mathcal{S}$ and $g \in G$.
We say that the family of polynomials $\mathfrak{q} = (q_\mathbf{k})_{\mathbf{k} \in \mathcal{S}}$ is \emph{$G$-invariant} if 
$$
q_{g\mathbf k}=g q_{\mathbf k}
$$
for all $g \in G$ and $\mathbf{k}\in\mathcal S$. 
This equation can be spelled out as 
$
q_{g\mathbf{k}}(\mathbf{x}^{[0]}, \ldots, \mathbf{x}^{[n]}) = q_{\mathbf{k}}(\mathbf{x}^{[g0]}, \ldots, \mathbf{x}^{[gn]}).
$ 

Now, if $\mathfrak{q}$ is $G$-invariant, the resulting sos polynomial
$$
p = \sum_{\mathbf{k} \in \mathcal{S}} q_{\mathbf{k}}^2
$$ 
is
also $G$-invariant (since  $\mathbf{k} \mapsto g\mathbf{k}$ is a bijection on $\mathcal{S}$).  
In Theorem \ref{thm:sos-dec} \ref{thm:sos-dec:i}, we will prove the reverse direction, namely that 
every $G$-invariant sos polynomial $p$ has a $G$-invariant family  of polynomials $\mathfrak{q}$. 

To prove this result, we leverage a correspondence between matrices and polynomials given by the \emph{Gram map} $\mathcal{G}$  (similarly to Example \ref{ex:sepPoly}). For simplicity, we assume for the rest of this section that every local polynomial space uses the same number of variables, i.e.\
$$
\mathcal{P} = \mathbb{R}[\mathbf{x}^{[0]}] \otimes \cdots \otimes \mathbb{R}[\mathbf{x}^{[n]}]
$$
where $\mathbf{x}^{[i]} = (x^{[i]}_1, \ldots, x^{[i]}_m)$ for each $i \in [n]$.
Now consider a polynomial $p \in \mathcal{P}$ with $\deg_{\mathrm{loc}}(p) \leq 2d$. 
We can represent $p$ via the Gram map 
\begin{align*}
\mathcal{G}:   \mathcal{M}_D^{\otimes n+1} &\to \mathcal{P} \\
M & \mapsto \mathfrak{m}_{n,d}^t M \mathfrak{m}_{n,d}
\end{align*}
where 
$\mathfrak{m}_{n,d} = \mathfrak{m}_{d}(\mathbf{x}^{[0]}) \otimes \cdots \otimes \mathfrak{m}_{d}(\mathbf{x}^{[n]}) $
and we define $\mathfrak{m}_{d}(\mathbf{x})$ to be the monomial basis in $\mathbf{x}$ consisting of all monomials of degree at most $d$. 
In addition, $\mathcal{M}_{D}$ is the space of real matrices of size $D\times D$, where 
$D= \binom{m+d}{d}$. Note that $D$ is also the number of monomials in $m$ variables of degree at most $d$. 
We say that the matrix 
$M =\sum_{j=1}^{N} M^{[0]}_j \otimes \cdots \otimes M^{[n]}_j $ is  $G$-invariant  if 
$$
gM \coloneqq \sum_{j=1}^{N} M^{[g^{-1}0]}_j \otimes \cdots \otimes M^{[g^{-1}n]}_j =  M
$$
for every $g \in G$, that is, if $M$  is invariant with respect to all permutations of the tensor factors induced by the group action of $G$ on $[n]$.

\begin{lemma}[Gram matrix of invariant sos polynomials]
\label{lem:gram}
Let $p \in \mathcal{P}$ with $\deg_{\mathrm{loc}}(p) \leq 2d$. The following are equivalent:
\begin{enumerate}[label=(\roman*),ref=(\roman*),leftmargin=*] 
	\item \label{im:gram:i}
	$p$ is sos and $G$-invariant. 
	\item \label{im:gram:ii}
	There exists an $M \in  \mathcal{M}_D^{\otimes n+1}$  which is positive semidefinite and $G$-invariant such that $\mathcal{G}(M) = p.$
\end{enumerate}
\end{lemma}

\begin{proof} \ref{im:gram:ii} $\Rightarrow$ \ref{im:gram:i}.  If there exists such an $M$, since it is positive semidefinite, it has a rank decomposition $M = \sum_{k} v_k v_k^t$ 
where $v_k \in \left(\mathbb{R}^{D}\right)^{\otimes n+1}$. 
This   gives rise to a sos decomposition of $p$ via $\mathcal{G}$. Furthermore, since $gM = M$ for all $g \in G$, we obtain
\begin{align*}  gp&= p(\mathbf{x}^{[g0]}, \ldots, \mathbf{x}^{[gn]}) \\ &= \nn \mathfrak{m}_{d}(\mathbf{x}^{[g0]})^t \otimes \cdots \otimes \mathfrak{m}_{d}(\mathbf{x}^{[gn]})^t \cdot M \cdot \mathfrak{m}_{d}(\mathbf{x}^{[g0]}) \otimes \cdots \otimes \mathfrak{m}_{d}(\mathbf{x}^{[gn]}) \\  &= \mathfrak{m}_{d}(\mathbf{x}^{[g0]})^t \otimes \cdots \otimes \mathfrak{m}_{d}(\mathbf{x}^{[gn]})^t \cdot g^{-1} M \cdot \mathfrak{m}_{d}(\mathbf{x}^{[g0]}) \otimes \cdots \otimes \mathfrak{m}_{d}(\mathbf{x}^{[gn]}) 
\\  &= p\end{align*}
where the second equality holds by the $G$-invariance of $M$, and the last equality by the commutativity of polynomial multiplication.

\ref{im:gram:i} $\Rightarrow$ \ref{im:gram:ii}. Assume that $p = \sum_{k=1}^{N} q_k^2$ is $G$-invariant. Define $v_k \in \left(\mathbb{R}^{D}\right)^{\otimes n+1}$ such that $q_k = v_k^t \mathfrak{m}_{n,d}$ defines a positive semidefinite matrix $M' = \sum_{k=1}^{N} v_k v_k^t$ with $\mathcal{G}(M') = p$, where  $M'$ need not be $G$-invariant. By the $G$-invariance of $p$, we additionally have that $\mathcal{G}(gM') = p$ for every $g \in G$. Defining $M$ as the average
$$
M = \frac{1}{\vert G \vert} \sum_{g \in G} gM'
$$
we obtain a $G$-invariant and positive semidefinite matrix $M$. 
By linearity of the Gram map, we have that $\mathcal{G}(M) = p$.
\end{proof}

\begin{remark}[Gram matrix of invariant separable polynomials]\label{rem:sep}
A similar version of Lemma \ref{lem:gram} relates  invariant separable polynomials  
$p \in \mathcal{C}_{\mathrm{sep}} = \mathcal{C}^{[0]}_{\mathrm{sos}} \otimes \cdots \otimes \mathcal{C}^{[n]}_{\mathrm{sos}}$ with invariant separable matrices $M$. The only difference  is that the vectors $v_k$  should be  elementary tensors factors.
 \demo\end{remark}

In order to state and prove the main result of this section (Theorem \ref{thm:sos-dec}), it only remains to define invariant sos---this is the non-stringent version advocated above.

\begin{definition}[Invariant sos decompositions]
\label{def:sos-dec}
Let $G$ act on the weighted simplicial complex $\Omega$,
and let $\mathfrak{q} = (q_{\mathbf{k}})_{\mathbf{k} \in \mathcal{S}}$ be a family of polynomials.
\begin{enumerate}[label=(\roman*),ref=(\roman*),leftmargin=*]
\item \label{def:sos-dec:i}
An \emph{$(\Omega,G)$-decomposition of the family $\mathfrak{q}$} is a decomposition 
$$q_{\mathbf{k}}= \sum_{\alpha \in \mathcal{I}^{\widetilde{\mathcal{F}}}} q_{k_0, \alpha_{\vert_0}}^{[0]}(\mathbf{x}^{[0]}) \cdots q^{[n]}_{k_n, \alpha_{\vert_n}}(\mathbf{x}^{[n]})$$
for every $\mathbf{k} \in \mathcal{S}$, where 
$$ q_{k_i, \beta}^{[i]} \in \mathbb{R}[\mathbf{x}^{[i]}]$$
and 
$$ q_{k_{i}, \beta}^{[i]}(\mathbf{x}^{[i]}) = q_{k_{i}, {}^g \beta}^{[gi]}(\mathbf{x}^{[i]})$$
for every $i \in [n]$, $\beta \in \mathcal{I}^{\widetilde{\mathcal{F}}_i}$, $g \in G$ and $\mathbf{k} \in \mathcal{S}$.
The smallest cardinality of $\mathcal{I}$ among all $(\Omega,G)$-decompositions is called the \emph{$(\Omega,G)$-rank} of $\mathfrak{q}$, denoted   $\rank_{(\Omega,G)}(\mathfrak{q}).$

\item  \label{def:sos-dec:ii}
An \emph{sos $(\Omega,G)$-decomposition} of $p\in\mathcal P$ is given by a sos decomposition into a family $\mathfrak{q}$ (that is, $p=\sum_{\mathbf k\in\mathcal S} q_{\mathbf k}^2$), together with an $(\Omega,G)$-decomposition of $\mathfrak{q}$.  
The minimal $(\Omega,G)$-rank among all such sos decompositions is called the \emph{sos $(\Omega,G)$-rank} of $p$, denoted $\sosrank_{(\Omega,G)}(p)$.
If $G$ is the trivial group action, we call the sos $(\Omega,G)$-decomposition just \emph{sos $\Omega$-decomposition} and denote its rank by $\sosrank_{\Omega}$. 
\end{enumerate}
\end{definition}

We are now ready to prove the main result regarding the existence of invariant sos polynomials: 
Every $G$-invariant sos polynomial $p$ has a $G$-invariant family $\mathfrak{q}$  (Theorem \ref{thm:sos-dec} \ref{thm:sos-dec:i}),
and  $\mathfrak q$ has an $(\Omega,G)$-decomposition if   $G$ is a free group action on $\Omega$ (Theorem \ref{thm:sos-dec} \ref{thm:sos-dec:ii}). 
The idea of the proof of Theorem \ref{thm:sos-dec} \ref{thm:sos-dec:i}   is to define $\mathfrak{q}$ as the square root of $p$, and show that this square root is also $G$-invariant.  
Some ideas of the proof are illustrated in Example \ref{ex:inv-sos}.

\begin{theorem}[Invariant sos decompositions]\label{thm:sos-dec} \quad\nobreakpar
\begin{enumerate}[label=(\roman*),ref=(\roman*),leftmargin=*]
\item \label{thm:sos-dec:i}
Let $p$ be a $G$-invariant sos polynomial. 
Then there exists   a $G$-invariant family of polynomials $\mathfrak{q} = (q_\mathbf{k})_{\mathbf{k} \in \mathcal{S}}$ such that $p = \sum_{\mathbf{k} \in \mathcal{S}} q_\mathbf{k}^2.$
Moreover, every element $q_{\mathbf{k}}$ admits a decomposition in which the local polynomials at site $i$ only depend on $k_i$, namely 
$$ 
q_{\mathbf{k}} = \sum_{j \in \mathcal{\mathcal{I}}} q_{{k_0}, j}^{[0]}(\mathbf{x}^{[0]}) \cdots q_{{k_n}, j}^{[n]}(\mathbf{x}^{[n]}). 
$$

\item \label{thm:sos-dec:ii}
Let $\Omega$ be a connected weighted simplicial complex with a free group action from  $G$.  
Then $\mathfrak q$ has an $(\Omega,G)$-decomposition, i.e.\ $\rank_{(\Omega,G)}(\mathfrak{q}) < \infty.$
\end{enumerate}
\end{theorem}

\begin{proof} 
\ref{thm:sos-dec:i}
We denote the monomial  $\mathbf{x} = (x_1, \ldots, x_m)$  with exponent  $\alpha = (\alpha_1,\ldots,\alpha_m)$ by 
$\mathbf{x}^\alpha = x_1^{\alpha} \cdot x_2^{\alpha_2} \cdots x_m^{\alpha_m}.$
Without loss of generality we can assume that $\deg_{\mathrm{loc}}(p) \leq 2d$. 
Define 
$$
\mathcal{S}_i = \big\{k \in \mathbb{N}^m: \vert k\vert \leq d\big\}
$$  
 and $\mathcal{S} = \mathcal{S}_0 \times \cdots \times \mathcal{S}_n.$  Note that $\mathcal S$ can be identified with the set of monomials in $\mathcal P$ of local degree at most  $d$ via the correspondence
$$
\mathcal{S} \to \mathcal{P}_d \colon \quad  \mathbf{k} \mapsto \mathbf{x}^\mathbf{k} \coloneqq \left(\mathbf{x}^{[0]}\right)^{k_0} \cdots \left(\mathbf{x}^{[n]}\right)^{k_n}.
$$
Note also that the permutations of variables $\mathbf{x}^{[i]} \mapsto \mathbf{x}^{[gi]}$  coincide with the group action  of $G$ on $\mathcal S,$ since
\be \label{eq:symrel} 
\left(\mathbf{x}^{[g0]}\right)^{k_0} \cdots \left(\mathbf{x}^{[gn]}\right)^{k_n} = \left(\mathbf{x}^{[0]}\right)^{k_{g^{-1}0}} \cdots \left(\mathbf{x}^{[n]}\right)^{k_{g^{-1}n}}.
\ee

Since $p$ is $G$-invariant and sos, by Lemma \ref{lem:gram} there exists a positive semidefinite and $G$-invariant matrix $M $ such that $\mathcal{G}(M) = p$. 
Now let $B$ be the (unique) positive semidefinite square root of $M$, i.e.\ $M = B^2$. 
Since $M$ is a matrix, $B$ admits a polynomial expression in $M$ and hence $B$ is also $G$-invariant. 
Define the polynomials $q_{\mathbf{k}}$ as 
$$
q_{\mathbf{k}} = \sum_{\mathbf{k'} \in \mathcal{S}} B_{\mathbf{k}, \mathbf{k'}} \left(\mathbf{x}^{[0]}\right)^{k'_{0}} \cdots \left(\mathbf{x}^{[n]}\right)^{k'_{n}}
$$
for $\mathbf{k} \in \mathcal{S}$. 
The family $\mathfrak{q} = (q_{\mathbf{k}})_{\mathbf{k} \in \mathcal{S}}$ is $G$-invariant, since
\begin{align*} \nn 
gq_k
&= \sum_{\mathbf{k'} \in \mathcal{S}} B_{g \mathbf{k}, g \mathbf{k'}} \left(\mathbf{x}^{[0]}\right)^{k'_{g^{-1}0}} \cdots \left(\mathbf{x}^{[n]}\right)^{k'_{g^{-1}n}}\\
&= \sum_{\mathbf{k'} \in \mathcal{S}} B_{g \mathbf{k}, \mathbf{k'}} \left(\mathbf{x}^{[0]}\right)^{k'_0} \cdots \left(\mathbf{x}^{[n]}\right)^{k'_n} = q_{g \mathbf{k}} 
\end{align*}
where we have used the fact that $B_{\mathbf{k}, \mathbf{k'}} = B_{g \mathbf{k}, g \mathbf{k'}}$ for every $g \in G$ (which is just the $G$-invariance of $B$),  together with Equation \eqref{eq:symrel}  and  bijectivity of the map  $\mathbf{k'} \mapsto g \mathbf{k'}$.  In addition,
$$ 
\sum_{\mathbf{k} \in \mathcal{S}} q_{\mathbf{k}}^2 = \mathfrak{m}_{n,d}^t B^t B \mathfrak{m}_{n,d} = \mathcal{G}(M) = p 
$$
since $B^t B = B^2 = M$.
Moreover, $B$ admits a tensor decomposition
$$B_{\mathbf{k}, \mathbf{k'}} = \sum_{j \in \mathcal{I}} \left(B_j^{[0]}\right)_{k_0, k_0'} \cdots \left(B_j^{[n]}\right)_{k_n, k_n'}.$$
Using the definition of $q_{\mathbf{k}}$ leads to the last statement of \ref{thm:sos-dec:i}.

\ref{thm:sos-dec:ii} 
The proof is similar to that of Theorem \ref{thm:omegaG-dec}. Start with decompositions
$$q_{\mathbf k} = \sum_{j \in \mathcal{I}} q_{k_0, j}^{[0]}(\mathbf{x}^{[0]}) \cdots q_{k_n, j}^{[n]}(\mathbf{x}^{[n]})$$
for every $\mathbf{k} = (k_0, \ldots, k_n) \in \mathcal{S}$.  
From the construction of Theorem \ref{thm:omega-dec} it follows that  every polynomial $q_{\mathbf{k}}$ has a decomposition of the form
$$
q_{\mathbf{k}} = \sum_{\alpha \in \mathcal{I}^{\widetilde{\mathcal{F}}}} p_{k_0, \alpha_{\vert_0}}^{[0]}(\mathbf{x}^{[0]}) \cdots p_{k_n, \alpha_{\vert_n}}^{[n]}(\mathbf{x}^{[n]})
$$
where $\widetilde{\mathcal{F}}$ is the  set of facets of $\Omega$.
We now construct a decomposition for every $q_{\mathbf{k}}$ which additionally satisfies the symmetry conditions of Definition \ref{def:sos-dec} \ref{def:sos-dec:i}. 
Since $G$ is free, by Remark \ref{rem:grac} \ref{rem:grac:vi}, 
there exists a $G$-linear map $\mathbf{z}: \widetilde{\mathcal{F}} \to G$. We consider the new index set $\hat{\mathcal{I}} \coloneqq \mathcal{I} \times G$, together with the projection maps $\pi_1: \hat{\mathcal{I}} \to \mathcal{I}$ and $\pi_2: \hat{\mathcal{I}} \to G$. For each $i \in [n]$ and $\beta \in \hat{\mathcal{I}}^{\mathcal{\widetilde{F}}_i}$ we  define the following local polynomials
$$
q_{k_i,\beta}^{[i]}(\mathbf{x}^{[i]}) \coloneqq \left\{ \begin{array}{ll} p^{[gi]}_{k_i, {}^g (\pi_1 \circ \beta)}(\mathbf{x}^{[i]}) & : \pi_2 \circ \beta = ({}^{g^{-1}} \mathbf{z})_{\vert_i} \\0 &: \textrm{else.} \end{array} \right.
$$
Similarly to the discussion in the proof of Theorem \ref{thm:omegaG-dec} we see that
$$ q_{k_i,{}^{g}\beta}^{[gi]}(\mathbf{x}^{[i]}) = q_{k_i,\beta}^{[i]}(\mathbf{x}^{[i]})$$
and
$$ \vert G\vert  \cdot q_{\mathbf{k}}= \sum_{\hat{\alpha} \in \mathcal{\hat{I}}^{\widetilde{\mathcal{F}}}} q_{k_0, \hat{\alpha}_{\vert_0}}^{[0]}(\mathbf{x}^{[0]}) \cdots q_{k_n, \hat{\alpha}_{\vert_n}}^{[n]}(\mathbf{x}^{[n]})$$
holds for every $\mathbf{k} \in \mathcal{S}$. But this implies the existence of an $(\Omega,G)$-decomposition of $\mathfrak{q}.$  
\end{proof}

\begin{remark}[More general version of Theorem \ref{thm:sos-dec} \ref{thm:sos-dec:i}] \phantom{*}
\label{rem:semiinvdec} 
\begin{enumerate}[label=(\roman*)]
	\item In \cite[Theorem 5.3]{Ga02}, the authors prove the existence of so-called semi-symmetric sos decompositions for general representations of finite groups, by using Schur's lemma on the Gram matrix. 
Theorem \ref{thm:sos-dec} \ref{thm:sos-dec:i} is weaker than that, as it only considers group actions that permute the tensor product spaces, but gives an elementary proof.

	\item There are also other characterizations of invariant sum-of-squares decompositions, like \cite[Corollary 2.7]{De21d}. Our decompositions are really sums-of-square decompositions. To highlight the difference to our framework, let us consider a decomposition of the polynomial $p = x^2 + y^2 + (x-y)^2$. According to Corollary 2.7 of Debus, Riener, $p$ decomposes into
$$ x^2 + y^2 + (x-y)^2 = \langle A, B(x,y) \rangle$$
with
$$A = \begin{pmatrix} 1 & -1 \\ -1 & 1 \end{pmatrix} \quad \text{ and } \quad B(x,y) = \mathcal{R}_{C_2}\left(\begin{pmatrix} x \\ y \end{pmatrix} \cdot \begin{pmatrix} x & y \end{pmatrix} \right) = \begin{pmatrix} \frac{x^2 + y^2}{2} & xy \\ xy & \frac{x^2 + y^2}{2} \end{pmatrix} $$
where $\mathcal{R}_G$ is the Reynolds operator applied to every entry of the matrix separately, i.e.\ $\mathcal{R}_{C_2}(x^2) = \frac{1}{2} (x^2 + y^2)$ and $\mathcal{R}_{C_2}(xy) = xy$.

In contrast, $p$ factorizes according to Theorem \ref{thm:sos-dec} in our paper as
$$p = p_{1,1}(x,y)^2 + p_{2,2}(x,y)^2 + p_{1,2}(x,y)^2 + p_{2,1}(x,y)^2$$
with $p_{1,1} = p_{2,2} = \frac{1}{2} (x-y)^2$ and $p_{1,2} = x^2$ and $p_{2,1} = y^2$.
\end{enumerate}
\demo
\end{remark}

From Theorem \ref{thm:sos-dec} it follows that: 

\begin{corollary}[Invariant sos polynomials with free group action]\label{cor:sos}
Let $\Omega$ be a connected weighted simplicial complex with a free group action from  $G$. Then every $p \in \mathcal{P}$ which is sos and $G$-invariant has an sos $(\Omega,G)$-decomposition, i.e.\ $\sosrank_{(\Omega,G)}(p) < \infty.$
\end{corollary}

\begin{example}[Illustrating invariant sos decompositions] \label{ex:inv-sos}
Consider again the polynomial from Example \ref{ex:omegaG-exp}, 
$$
p=x^2+y^2+4(1+xy)^2,
$$
which is  sos and invariant with respect to the  permutation of $x$ and $y$. We have already seen that ${\rm rank}_{(\Delta,C_2)}(p)=2.$ 
By a similar argument as in Example \ref{ex:sepPoly},  it can be shown that $p$ is not separable with respect to the local sos cones.

To obtain a sos decomposition  we   follow the proof of Theorem \ref{thm:sos-dec}. 
We obtain $\mathcal S=\{0,1\}\times\{0,1\}$ with $G=C_2$ permuting the entries of the tuples, and obtain a $C_2$-invariant sos decomposition of $p$ via the following family of polynomials: 
$$
q_{(0,0)}=q_{(1,1)}=\sqrt{2}(1+xy),\quad q_{(0,1)}=y,\quad  q_{(1,0)}=x.
$$ 

On the double edge $\Delta$ we obtain an $(\Delta, C_2)$-decomposition of the family   via the following family of polynomials
\begin{align*}q_0^{[0]}&
=\left(\begin{array}{ccc}\sqrt[4]{2} t & \frac{1}{\sqrt{2}} & 0 \\1 & 0 & 0 \\0 & 0 & 0\end{array}\right), \quad q_0^{[1]}={q_0^{[0]}}^t\\
q_1^{[0]}& 
=\left(\begin{array}{ccc}0 & 0 & 0 \\\sqrt{2}t & \sqrt[4]{2}t & 0 \\0 & 0 & \sqrt[4]{2}\end{array}\right), \quad q_1^{[1]}={q_1^{[0]}}^t. 
\end{align*} 
where the matrix notation denotes that the rows are indexed by $\alpha=1,2,3$ and the columns by $\beta=1,2,3$. 
This shows that 
$$\sosrank_{(\Delta,C_2)}(p)\leq \sosrank_{(\Delta,C_2)}(\mathfrak q) \leq 3.
$$

On the single edge $\Sigma_1$, 
a decomposition of  $\mathfrak q$  requires vectors $a,b,c,d\in\mathbb R^d$ of length $\sqrt[4]{2}$, with $a,b,c$ pairwise orthogonal, $d$ orthogonal to $b$ and $c$, and $\langle a,d\rangle =1$. This is provided by 
\begin{align*} 
q^{[0]}_0&=q_0^{[1]}
= \left(a_\alpha+b_\alpha t\right)_{\alpha=1,\ldots, d}\\ 
q^{[0]}_1&=q_1^{[1]}
= \left(c_\alpha+d_\alpha t\right)_{\alpha=1,\ldots, d}
\end{align*} 
where  $(\: )_{\alpha}$  denotes a vector indexed by $\alpha$. 
 Since such vectors can only be found in dimension $d\geq 4$,  we obtain 
$$\sosrank_{(\Lambda_1,C_2)}(p)\leq \sosrank_{(\Lambda_1,C_2)}(\mathfrak q) = 4.$$

We can also write $p$ as a sum of symmetric squares: 
$$
p=\left(2+\frac32xy\right)^2 + \left(x+y\right)^2+ \left(\sqrt{\frac74}xy\right)^2.
$$ 
We now reset the variables $\mathcal S_0=\mathcal S_1=\{1,2,3\}, \mathcal S=\mathcal S_1\times\mathcal S_2,$   as well as 
$$
q_{(1,1)}=2+\frac32xy,\ q_{(2,2)}=x+y,\ q_{(3,3)}=\sqrt{\frac74}xy,
$$ 
and all other $q_{\mathbf k}=0$. This gives rise to the $C_2$-invariant family $\mathfrak q= \left(q_{\mathbf k}\right)_{\mathbf k\in\mathcal S}$ that provides an sos decomposition of $p$ with 
$$
\sosrank_{(\Delta,C_2)}(\mathfrak q) \leq 3.
$$ 
But for  the single edge, there does not exist a decomposition for the family $\mathfrak q$. This is because already $q_{(2,2)}=x+y$ does not admit an $(\Lambda_1,C_2)$-decomposition (without a minus sign). So 
$$
\sosrank_{(\Lambda_1,C_2)}(\mathfrak q)=\infty.$$\demo
\end{example}

\section{Inequalities and separations between the ranks}
\label{sec:inequalitites}

In this section, we study rank inequalities (Section \ref{ssec:inequalities}),
provide an upper bound for the separable rank (Section \ref{ssec:upper}),
and show separations between ranks (Section \ref{ssec:sep}).

\subsection{Inequalities between ranks}
\label{ssec:inequalities}

In this section, we show three relations between the introduced ranks (Proposition \ref{pro:inequalitites}), which are similar to the statements established for tensor decompositions in \cite[Proposition 29]{De19d}. For the inequality between sos and separable decompositions we will need to assume that $(\Omega,G)$ is \emph{factorizable}: 

\begin{definition}[Factorizable] 
Let $\Omega$ be a weighted simplicial complex with a group action from  $G$. We say that $(\Omega,G)$ is   \emph{factorizable} if for each finite index set $\mathcal I$ the following system of equations admits a solution with all $C^{[i]}_{\beta} > 0$ and $C_{{}^g \beta}^{[gi]} = C_{\beta}^{[i]}$ for all $i \in [n]$,   $\beta \in \mathcal{I}^{\widetilde{\mathcal{F}}_i}$, and $g \in G$:
\be \label{eq:Eq-system} C_{\alpha_{\vert_0}}^{[0]} \cdot C_{\alpha_{\vert_1}}^{[1]} \cdots  C_{\alpha_{\vert_n}}^{[n]} = K_{\alpha}^{-1} \quad \textrm{ for all } \alpha \in \mathcal{I}^{\widetilde{\mathcal{F}}}, \ee 
where
$$
K_{\alpha} \coloneqq \big\vert \{\gamma \in \mathcal{I}^{\widetilde{\mathcal{F}}}: \exists g_0, \ldots, g_n \in G \textrm{ with } g_i i = i \textrm{ and } \left({}^{g_i}\gamma\right)_{\vert_i}= \alpha_{\vert_i} \mbox{ for all }  i\in[n]\}\big\vert. 
$$
\end{definition}

Note that Equation \eqref{eq:Eq-system} can be seen as a system of linear equations by taking the logarithm on the left and the right hand side.

All examples of group actions on a weighted simplicial complex $\Omega$  considered in this paper are factorizable, as the following example shows. 

\begin{example}[Factorizable group actions] \quad\nobreakpar
\begin{enumerate}[label=(\roman*),ref=(\roman*),leftmargin=*]
\item 
If $K_{\alpha} = 1$ for every $\alpha \in \mathcal{I}^{\widetilde{\mathcal{F}}}$, then $C_{\beta}^{[i]} = 1$ solves Equation \eqref{eq:Eq-system}. This in particular shows that $(\Omega,G)$ is factorizable whenever the action of $G$ on the vertices $[n]$ is free. In addition, this also implies that $(\Sigma_n, S_{n+1})$ is factorizable.

\item 
Let $\Omega = \Delta$ be the double edge and let $G = \{e, g\}$ act by keeping  the vertices  fixed (i.e.\ $e i = g i = i$) and  flipping the facets (i.e.\ $g \mathfrak{a} = \mathfrak{b}$, $g \mathfrak{b} = \mathfrak{a}$)\footnote{Note that this is different from the symmetric double edge of Example \ref{ex:group_EdgeDoubleEdge}, since here the vertices remain fixed. The usual action on the double edge  is free, and hence factorizable as well.}.
 In this situation, we have
$$K_{\alpha_1, \alpha_2} = \left\{\begin{array}{cl} 1 & \textrm{: if } \alpha_1 = \alpha_2\\ 2 & \textrm{: if } \alpha_1 \neq \alpha_2. \end{array}\right.$$
A  solution of Equation \eqref{eq:Eq-system} is given by
$$C_{\alpha_1, \alpha_2}^{[i]} = \left\{\begin{array}{cl} 1 & \textrm{: if } \alpha_1 = \alpha_2\\ 1/\sqrt{2} & \textrm{: if } \alpha_1 \neq \alpha_2. \end{array}\right.$$
Hence, $(\Delta,G)$ is also factorizable.\demo
\end{enumerate}
\end{example}

In fact, we are not aware of any  non-factorizable $(\Omega,G)$ structures, leading to the following open question.

\begin{question}
\label{qu:fact}
Are there non-factorizable $(\Omega,G)$ structures?
\end{question}

We are now ready to present the rank inequalities.

\begin{proposition}[Rank inequalities] \label{pro:inequalitites}
Let $p \in \mathcal{P}$.  
\begin{enumerate}[label=(\roman*),leftmargin=*]
	\item \label{pro:inequalitites:i}
	$\rank_{(\Omega,G)}(p) \leq \seprank_{(\Omega,G)}(p)$ for any  separable cone.
	\item \label{pro:inequalitites:ii}
	$\rank_{(\Omega,G)}(p) \leq \sosrank_{(\Omega,G)}(p)^2$.
	\item \label{pro:inequalitites:iii}
	If $(\Omega,G)$ is factorizable, then $$\sosrank_{(\Omega,G)}(p) \leq \seprank_{(\Omega,G)}(p)$$   for the separable cone over  local sos polynomials.
		\item \label{pro:inequalitites:iv}
		If the action of $G$ on the connected weighted simplicial complex  $\Omega$ is free and $(\Omega,G)$ is factorizable, then for every $G$-invariant $p\in\mathcal P$ we have	$$\sosrank_{(\Omega,G)}(p) \leq \vert G\vert  \cdot \seprank_{\Omega}(p),$$ for the separable cone over local  sos-polynomials.
\end{enumerate}
\end{proposition}

\begin{proof}
\ref{pro:inequalitites:i}. Clear, since every separable $(\Omega,G)$-decomposition is an $(\Omega,G)$-decomposition.

\ref{pro:inequalitites:ii}.  Let $\mathfrak{q} = (q_{\mathbf{k}})_{\mathbf{k} \in \mathcal{S}}$ be a $G$-invariant sos-decomposition of $p$, with  an $(\Omega,G)$-decomposition
$$q_{\mathbf{k}} = \sum_{\alpha \in \mathcal{I}^{\widetilde{\mathcal{F}}}} q_{k_0, \alpha_{\vert_0}}^{[0]}(\mathbf{x}^{[0]}) \cdots q_{k_n, \alpha_{\vert_n}}^{[n]}(\mathbf{x}^{[n]})$$ for each $\mathbf k\in\mathcal S=\mathcal S_0\times\cdots\times\mathcal S_n$.
Defining $\hat{\mathcal{I}} \coloneqq \mathcal{I} \times \mathcal{I}$ and
$$p_{\beta, \beta'}^{[i]} \coloneqq \sum_{k \in\mathcal S_i} q_{k, \beta}^{[i]}(\mathbf{x}^{[i]}) \cdot q_{k, \beta'}^{[i]}(\mathbf{x}^{[i]})$$
we obtain  a valid $(\Omega,G)$-decomposition of $p$, with $\rank_{(\Omega,G)}(p) \leq \vert\mathcal{\hat{I}}\vert = \vert\mathcal{I}\vert^2$:
\be \nn p &=& \sum_{\mathbf{k} \in \mathcal{S}} \sum_{\alpha, \alpha' \in \mathcal{I}^{\widetilde{\mathcal{F}}}} q_{k_0, \alpha_{\vert_0}}^{[0]}(\mathbf{x}^{[0]}) \cdot q_{k_0, \alpha'_{\vert_0}}^{[0]}(\mathbf{x}^{[0]}) \cdots q_{k_n, \alpha_{\vert_n}}^{[n]}(\mathbf{x}^{[n]}) \cdot q_{k_n, \alpha'_{\vert_n}}^{[n]}(\mathbf{x}^{[n]}) \\ \nn &=& \sum_{(\alpha,\alpha') \in \mathcal{\hat{I}}^{\widetilde{\mathcal{F}}}} p_{\alpha_{\vert_0}, \alpha'_{\vert_0}}^{[0]}(\mathbf{x}^{[0]}) \cdots p_{\alpha_{\vert_n}, \alpha'_{\vert_n}}^{[n]}(\mathbf{x}^{[n]}). \ee

\ref{pro:inequalitites:iii}.  Let $p^{[i]}_{\beta}  \in\mathcal{C}^{[i]}_{\rm sos}$ for $\beta \in \mathcal{I}^{\widetilde{\mathcal{F}}_i}$ and  $i \in [n]$ be local polynomials from a separable $(\Omega,G)$-decomposition of $p$. So there exist sos decompositions  
$$p^{[i]}_{\beta} = \sum_{k=1}^{N} \left(\tau_{k, \beta}^{[i]}\right)^2$$ with $\tau_{k,\beta}^{[i]}\in\mathbb R[x^{[i]}]$ (and we can clearly use the same sum length $N$ for all $i,\beta$).
We can in addition assume without loss of generality that
\be \nn \tau_{k, {}^g \beta}^{[gi]}(\mathbf{x}^{[i]}) = \tau_{k, \beta}^{[i]}(\mathbf{x}^{[i]}) \ee holds for all $i,\beta, k$ and $g$.
Indeed, just consider the action of $G$ on $$\bigcup_{i\in[n]}\{i\}\times \mathcal I^{\widetilde{\mathcal F}_i}$$ given by $g\cdot (i,\beta)\coloneqq (gi,{}^g\beta),$ and fix for every orbit precisely one representative $(i_1, \beta_1), \ldots, (i_M, \beta_M)$. Then choose one sos decomposition for each $p^{[i_\ell]}_{\beta_\ell}$ and use the same along its orbit. This works since we have $p^{[gi]}_{{}^g\beta}(x^{[i]})=p^{[i]}_\beta(x^{[i]})$ for all $i,\beta$ by assumption.

Now since $(\Omega,G)$ is factorizable, we can choose some positive and $G$-invariant solution $\left(C_\beta^{[i]}\right)_{\beta, i}$ of Equation \eqref{eq:Eq-system}. Using the above representatives $(i_\ell, \beta_\ell)$ again, we  now define
$$q^{[i]}_{(\ell, k), \beta} \coloneqq \left\{\begin{array}{c l} \sqrt{C^{[i]}_{\beta}} \cdot \tau_{k,\beta}^{[i]}(\mathbf{x}^{[i]}) & \textrm{: if } \exists g \in G: (i,\beta) = (g i_{\ell}, {}^g \beta_{\ell})\\[0.5cm]
0 & \textrm{: else}
\end{array} \right.$$
where $\ell \in \{1, \ldots, M\}$, $k \in \{1, \ldots, N\}$ and $\beta \in \mathcal{I}^{\widetilde{\mathcal{F}}_i}$.
By definition, we have $$q^{[gi]}_{(\ell,k), {}^g \beta}(\mathbf{x}^{[i]}) =  q^{[i]}_{(\ell,k), \beta}(\mathbf{x}^{[i]}),$$ and hence
$$q_{((\ell_0,k_0),\ldots,(\ell_n,k_n))}\coloneqq \sum_{\alpha \in \mathcal{I}^{\widetilde{\mathcal{F}}}} q^{[0]}_{(\ell_0,k_0), \alpha_{\vert_0}}(\mathbf{x}^{[0]}) \cdots q^{[n]}_{(\ell_n,k_n), \alpha_{\vert_n}}(\mathbf{x}^{[n]})$$
is a valid $(\Omega,G)$-decomposition of the $G$-invariant family
$$\mathfrak{q} \coloneqq \left(q_{((\ell_0,k_0),\ldots,(\ell_n,k_n))}\right)_{(\ell_i, k_i) \in \mathcal{S}_i}$$
where $\mathcal{S}_i = \{1, \ldots, M\} \times \{1, \ldots, N\}$. This family is also  an sos decomposition of $p$, since
\be \nn \sum_{\forall i: \: (\ell_i, k_i) \in \mathcal{S}_i}  q_{((\ell_0,k_0),\ldots,(\ell_n,k_n))}^2 &=& \sum_{\alpha \in \mathcal{I}^{\widetilde{\mathcal{F}}}} K_{\alpha} \cdot C_{\alpha_{\vert_0}}^{[0]} \cdots C_{\alpha_{\vert_n}}^{[n]} \cdot p^{[0]}_{\alpha_{\vert_0}}(\mathbf{x}^{[0]}) \cdots p^{[n]}_{\alpha_{\vert_n}}(\mathbf{x}^{[n]}) \\ \nn &=& p.
\ee
Here we have used Equation \eqref{eq:Eq-system}, as well as $G$-invariance of the $C^{[i]}_\beta$ and the $\tau_{k,\beta}^{[i]}.$

Finally, \ref{pro:inequalitites:iv}  follows  from  \ref{pro:inequalitites:iii} and Corollary \ref{cor:sep-ranks}:  
$$\sosrank_{(\Omega,G)}(p) \leq \seprank_{(\Omega,G)}(p)\leq \vert G\vert  \cdot \seprank_{\Omega}(p).$$
\end{proof}

\subsection{An upper bound for the separable rank}
\label{ssec:upper}

In this short section we provide an upper bound for the separable $(\Omega,G)$-rank with respect to the number of local variables $m_i$ and the polynomial's local degree. For simplicity, we again assume that all local polynomial spaces use the same number of variables, $m \coloneqq m_i = m_j$ for $i, j \in [n]$. For $p \in \mathcal{P}$ recall that the local degree of $p$, denoted  $\deg_{\mathrm{loc}}(p),$ is  the smallest integer $d \in \mathbb{N}$ such that
$$
p \in \mathbb{R}[\mathbf{x}^{[0]}]_d \otimes \cdots \otimes \mathbb{R}[\mathbf{x}^{[n]}]_d
$$
where $\mathbb{R}[\mathbf{x}]_d$ is the space of polynomials in variables $\mathbf{x}$ of degree at most $d$.

\begin{proposition}[Upper bound for separable rank]\label{pro:uppersep}
Let $p \in \mathcal{P}$ be separable and $G$-invariant, and let $\Omega$ be a connected weighted simplicial complex with a free group action from  $G$. Then
$$\seprank_{(\Omega,G)}(p) \leq \vert G\vert  \cdot \binom{\deg_{\mathrm{loc}}(p) + m}{\deg_{\mathrm{loc}}(p)}^{n+1}$$
for any separable cone.
\end{proposition}

\begin{proof}
Let $d =\deg_{\mathrm{loc}}(p)$. Then $p \in \mathbb{R}[\mathbf{x}^{[0]}]_d \otimes \cdots \otimes \mathbb{R}[\mathbf{x}^{[n]}]_d$.
Since 
$$\dim\left(\mathbb{R}[\mathbf{x}^{[i]}]_d\right) = \binom{d+m}{d}$$
for all $i \in [n]$, $p$ is a conic combination of at most $ \binom{d+m}{d}^{n+1}$ elementary products with factors from the local cones  by Carath\'eodory's Theorem (see for example \cite[Theorem 2.3]{Ba02}). From the proof of Theorem \ref{thm:omega-dec}, we have that 
$$
\seprank_{\Omega}(p) \leq \binom{d + m}{d}^{n+1}.
$$
The result now follows from Corollary \ref{cor:sep-ranks}.
\end{proof}

\subsection{Separations}
\label{ssec:sep}
Here we will show \emph{separations} between the ranks, which we will define shortly. 
Throughout this section we will consider separable decompositions only with respect to the local sos cones. 
We know from Proposition \ref{pro:inequalitites} that the separable rank upper bounds both the rank and sos-rank.  
Here we will show that a reverse inequality is impossible: in particular, there are no functions $f,g\colon \mathbb{N} \to \mathbb{N}$ such that
\be \nn \seprank_{\Lambda_1}(p) \leq f\big(\sosrank_{\Lambda_1}(p)\big) \quad \textrm{ and } \quad \sosrank_{\Lambda_1}(p) \leq g\big(\rank_{\Lambda_1}(p)\big)\ee
for all $m \in \mathbb{N}$ and polynomials $p \in \mathbb{R}[\mathbf{x}^{[0]}, \mathbf{x}^{[1]}]$ with $\mathbf{x}^{[i]} \coloneqq (x^{[i]}_1, \ldots, x^{[i]}_m)$. 
This  is called a \emph{separation} between $\sosrank$ and $\seprank$, or $\rank$ and $\sosrank$, respectively. 
We prove the separations by a reduction to matrix factorizations of entrywise nonnegative matrices, which themselves exhibit separations \cite{Fa14, Go12}.

For this reason, we focus on the subspace of $(n+1)$-quadratic forms in $\mathcal{P}$ and relate it with tensors. For $T \in  \mathbb{R}^m \otimes \cdots \otimes \mathbb{R}^m $ we define the polynomial
\be\label{eq:pT}
p_T \coloneqq  \sum_{j_0, \ldots, j_n = 1}^{m} T_{j_0, \ldots, j_n} \left(x^{[0]}_{j_0} \right)^2 \cdots \left(x^{[n]}_{j_n} \right)^2\in\mathbb R[x^{[0]},\ldots, x^{[n]}].
\ee
There is a one-to-one correspondence between the tensor $T$ and the polynomial $p_T$. 
In addition, entrywise nonnegativity of $T$ fully characterises the 
nonnegativity and the sos property of $p_T$: 

\begin{lemma}[Positivity correspondence between tensors and polynomials]
\label{lem:correspondencePositivity}
The map
\begin{align*}
\mathbb{R}^m \otimes \cdots \otimes \mathbb{R}^m &\to \mathcal{P}\\ T &\mapsto p_T
\end{align*}
(where $p_T$ is given in Equation \eqref{eq:pT}) is linear and injective. In addition, the following  are equivalent:
\begin{enumerate}[label=(\roman*),ref=(\roman*),leftmargin=*]
	\item \label{lem:correspondencePositivity:i}
	$T$ is entrywise nonnegative.
	\item  \label{lem:correspondencePositivity:ii}
	$p_T$ is a sum of squares.
	\item  \label{lem:correspondencePositivity:iii}
	$p_T$ is globally nonnegative (as a polynomial function).
\end{enumerate}
\end{lemma}

\begin{proof}
Linearity and injectivity are immediate (each entry of $T$ clearly gives rise to a different monomial). 

The implications \ref{lem:correspondencePositivity:i} $\Rightarrow$ \ref{lem:correspondencePositivity:ii} $\Rightarrow$ \ref{lem:correspondencePositivity:iii} are clear, since a nonnegative tensor $T$ generates a sum of squares,  since every sum of squares  is globally nonnegative. For \ref{lem:correspondencePositivity:iii} $\Rightarrow$ \ref{lem:correspondencePositivity:i} assume that $T$ is not nonnegative, so there exist $j_0, \ldots, j_n$ such that
$T_{j_0, \ldots, j_n} < 0$. Then
$$
p(e_{j_0}, \ldots, e_{j_n}) = T_{j_0, \ldots, j_n} < 0,$$
which shows that $p$ is not nonnegative.
\end{proof}

In order to ``borrow'' the separations of tensor decompositions to derive separations of polynomial decompositions, we now define decompositions of tensors, which were introduced  in \cite{De19d}.

\begin{definition}[Invariant decompositions of tensors \cite{De19d}]\label{def:tensor} 
Let $T\in \mathbb{R}^{m}\otimes \cdots \otimes \mathbb{R}^{m}$.
\begin{enumerate}[label=(\roman*),ref=(\roman*),leftmargin=*]
\item \label{def:tensor:i}
An \emph{$(\Omega,G)$-decomposition of $T$} is given by families
$$\mathcal{T}^{[i]} = \left(T^{[i]}_{\beta} \right)_{\beta \in \mathcal{I}^{\widetilde{\mathcal{F}}_i}}$$
where $T^{[i]}_{\beta} \in \mathbb{R}^{d}$ for all $i \in [n]$ and $\beta \in \mathcal{I}^{\mathcal{\widetilde{F}}_i},$ such that
$$ T = \sum_{\alpha \in \mathcal{I}^{\widetilde{\mathcal{F}}}} T^{[0]}_{\alpha_{\vert_0}} \otimes T^{[1]}_{\alpha_{\vert_1}} \otimes \cdots \otimes T^{[n]}_{\alpha_{\vert_n}} $$
and
$$T^{[gi]}_{{}^g \beta} = T^{[i]}_{\beta}$$
for all $i \in [n]$ and $\beta \in \mathcal{I}^{\widetilde{\mathcal{F}}_i}$. The minimal cardinality of $\mathcal{I}$ among all $(\Omega,G)$-decompositions is called the \emph{$(\Omega,G)$-rank} of $T$, denoted  $\rank_{(\Omega,G)}(T)$. 
\item 
A \emph{nonnegative $(\Omega,G)$-decomposition of $T$} is an $(\Omega,G)$-decomposition of $T$ where all local vectors $T_{\beta}^{[i]}$ have nonnegative entries. 
The corresponding rank is called the \emph{nonnegative $(\Omega,G)$-rank} of $T$, denoted  $\nnrank_{(\Omega,G)}(T).$
\item 
A \emph{positive semidefinite $(\Omega,G)$-decomposition of $T$} consists of positive semidefinite matrices (indexed by $\beta, \beta' \in \mathcal{I}^{\widetilde{\mathcal{F}}_i}$) $$ E_j^{[i]}\in \mathcal{M}_{\mathcal I^{\widetilde{\mathcal F}_i}}^{+}(\mathbb{R})$$ 
 for $i \in [n]$ and $j \in \{1,\ldots, m\}$ such that  $$\left(E_j^{[gi]}\right)_{{}^g\beta, {}^g\beta'}=\left(E_j^{[i]}\right)_{\beta,\beta'}$$ for all $i,g,j,\beta,\beta'$, and $$T_{j_0, \ldots, j_n}=\sum_{\alpha,\alpha'\in\mathcal I^{\widetilde{\mathcal F}}} \left(E_{j_0}^{[0]}\right)_{\alpha_{\mid_0}, \alpha'_{\mid_0}} \cdots \left(E_{j_n}^{[n]}\right)_{\alpha_{\mid_n}, \alpha'_{\mid_n}}$$ for all $j_0, \ldots, j_n$. 
The smallest cardinality of $\mathcal{I}$ among all positive semidefinite $(\Omega,G)$-decompositions is called the \emph{positive semidefinite $(\Omega,G)$-rank} of $T$, denoted ${\psdrank}_{(\Omega,G)}(T).$
\end{enumerate} 
\end{definition}

Note that in Ref.\ \cite{De19d}   all $(\Omega,G)$-decompositions are defined over complex numbers, whereas here we use real decompositions. 
This is however irrelevant for the purposes of Proposition \ref{pro:OmegaG-corr-tensors}. 
 
Every notion of invariant decomposition  of a tensor  $T$ can be associated to a notion of invariant decompositions of the corresponding polynomial $p_T,$ as we will show in the following proposition.

\begin{proposition}[Rank correspondence between tensors and polynomials]
\label{pro:OmegaG-corr-tensors}
Let $T \in  \mathbb{R}^m \otimes \cdots \otimes \mathbb{R}^m$ and the polynomial $p_T$ be given by Equation \eqref{eq:pT}. 
\begin{enumerate}[label=(\roman*),ref=(\roman*),leftmargin=*]
	\item \label{pro:OmegaG-corr-tensors:i}
	$\rank_{(\Omega,G)}(T)=\rank_{(\Omega,G)}(p_T).$
	\item \label{pro:OmegaG-corr-tensors:ii}
	 $\nnrank_{(\Omega,G)}(T)=\seprank_{(\Omega,G)}(p_T).$
	\item \label{pro:OmegaG-corr-tensors:iii}
	  $\psdrank_{(\Omega,G)}(T) \leq \sosrank_{(\Omega,G)}(p_T)$ with equality if $G$ acts freely on $[n]$.
\end{enumerate}
\end{proposition}

\begin{proof}
\ref{pro:OmegaG-corr-tensors:i}. Let the families $\left(T^{[i]}_{\beta} \right)_{\beta \in \mathcal{I}^{\widetilde{\mathcal{F}}_i}}$ provide an  $(\Omega,G)$-decomposition of $T$ as in Definition \ref{def:tensor} \ref{def:tensor:i}.  
Now consider the families 
$$
\mathcal P^{[i]}\coloneqq \left(\Psi_{T^{[i]}_{\beta}}(\mathbf x^{[i]})\right)_{\beta\in \mathcal{I}^{\widetilde{\mathcal{F}}_i}}
$$
where for a vector  $V \in \mathbb{R}^m$  the $\Psi$ notation indicates $\Psi_V(\mathbf x) \coloneqq \sum_{j=1}^{m} V_j x_j^2$. 
It is immediate to see that these families 
provide an $(\Omega,G)$-decomposition of $p_T$, using the same index set $\mathcal I$.
Conversely, observe that every $(\Omega,G)$-decomposition of $p_T$ consists without loss of generality  of local polynomials of the form
$$p_{\beta}^{[i]}= \sum_{j =1}^{m} \left(T_{\beta}^{[i]}\right)_{j} \left(x_{j}^{[i]}\right)^2$$
for certain $T_{\beta}^{[i]} \in \mathbb{R}^m$. All other possible monomials will have to cancel out in the total product and sum, and can therefore be omitted.
Thus the $T_{\beta}^{[i]}$ give rise to an $(\Omega,G)$-decomposition of $T$, again with the same index set $\mathcal I$.

Statement \ref{pro:OmegaG-corr-tensors:ii}  is proven exactly as \ref{pro:OmegaG-corr-tensors:i}, and using the fact that the local polynomials of an sos ($\Omega,G)$-decomposition  of $p_T$ must all be of degree $2$,  and thus have nonnegative coefficients at all the $\left(x_j^{[i]}\right)^2$.

For \ref{pro:OmegaG-corr-tensors:iii} we start with an sos $(\Omega,G)$-decomposition of $p_T$, where every local polynomial $q^{[i]}_{k, \beta}$ can (for degree reasons) be assumed to  be of the form
$$q^{[i]}_{k, \beta} = \sum_{j = 1}^m \left(B^{[i]}_j\right)_{k, \beta} x_j^{[i]}.$$  Now the matrices 
$$
E_j^{[i]} \coloneqq  \left(B_{j}^{[i]}\right)^t \left(B_{j}^{[i]}\right) \geqslant 0
$$ 
give rise to  a  positive semidefinite $(\Omega,G)$-decomposition of $T$ of the same rank as the initial decomposition. This can easily be seen by computing   the coefficient of $p_T$ at each monomial $(x_{j_0}^{[0]})^2\cdots (x_{j_n}^{[n]})^2$, and checking that it arises from the sos $(\Omega,G)$-decomposition.

For the reverse inequality, we assume that $G$ acts freely on $[n]$. We start with a positive semidefinite $(\Omega,G)$-decomposition of $T$, i.e.\
$$T_{j_0, \ldots, j_n} = \sum_{\alpha, \alpha' \in \mathcal{I}^{\widetilde{\mathcal{F}}}} \left(E_{j_0}^{[0]}\right)_{\alpha_{\vert_0}, \alpha'_{\vert_0}} \cdots \left(E_{j_n}^{[n]}\right)_{\alpha_{\vert_n}, \alpha'_{\vert_n}}$$
where all $E_{j}^{[i]}$ are positive semidefinite. Decompose $E_{j}^{[i]} = \left(B_{j}^{[i]}\right)^t \left(B_{j}^{[i]}\right)$ with the additional constraint that
$$
\left(B_{j}^{[gi]}\right)_{k, {}^g \beta} = \left(B_{j}^{[i]}\right)_{k, \beta}.
$$ 
Since $G$ acts freely on $[n],$ we can just choose certain $B^{[i]}_j$ and define the $B^{[gi]}_j$ along the orbit by that formula.
Now defining
$$ 
q^{[i]}_{(j,k), \beta} \coloneqq \left(\widetilde{B}_{j}^{[i]}\right)_{k, \beta} x_{j}^{[i]} 
$$
leads to a sos $(\Omega,G)$-decomposition of $p_T$ with $\sosrank_{(\Omega,G)}(p_T) \leq \vert\mathcal{I}\vert$. 
\end{proof}

\begin{remark}[The importance of being free]
The proof of Proposition \ref{pro:OmegaG-corr-tensors} \ref{pro:OmegaG-corr-tensors:iii} does not work in reverse direction if we do not assume that $G$ acts freely on $[n]$. Assume there exists $e \neq g \in G$ and $i \in [n]$ such that $gi = i$. Then, the construction into a symmetric factorization $\widetilde{B}_j^{[i]}$ implies that
$$\left(E_j^{[gi]}\right)_{{}^g \beta, \beta'} = \left(\widetilde{B}_j^{[gi]}\right)^{t}_{{}^g \beta, -} \left(\widetilde{B}_j^{[gi]}\right)_{-, \beta'} = \left(\widetilde{B}_j^{[i]}\right)^{t}_{\beta, -} \left(\widetilde{B}_j^{[i]}\right)_{-, \beta'} = \left(E_j^{[i]}\right)_{\beta, \beta'}$$
which is stronger than the symmetry of $E_j^{[i]}$ given in a positive semidefinite $(\Omega,G)$-decomposition.
\demo
\end{remark}

We now show that there is a separation between the ranks already for decompositions on the single edge. (Note that in the following corollary $p_m$ is a polynomial on the single edge). 

\begin{corollary}[Rank separations on the single edge] \label{cor:rank-sep} \quad\nobreakpar
Let $ p_m \in \mathbb{R}[x^{[0]}_1, \ldots ,x^{[0]}_m, x^{[1]}_1, \ldots, x^{[1]}_m]$. 
 \begin{enumerate}[label=(\roman*),label=(\roman*),leftmargin=*]
	\item \label{cor:rank-sep:i}
	There exists   a sequence of polynomials  $(p_m)_{m \in \mathbb{N}}$   such that 
	$$
	\rank_{\Lambda_1}(p_m) = 3 ,  \quad \sosrank_{\Lambda_1}(p_m) = 2  \quad \textrm{ and } \log_2(m)\leq \seprank_{\Lambda_1}(p_m) <\infty.
	$$

	\item \label{cor:rank-sep:ii}
			There exists  a sequence of polynomials  $(p_m)_{m \in \mathbb{N}}$  such that 
	$\rank_{\Lambda_1}(p_m) = 3$ and 
	$$
	\lim_{m \to \infty}  \sosrank_{\Lambda_1}(p_m)  = \infty 
	$$
	(where of course $\sosrank_{\Lambda_1}(p_m) <\infty$).
\end{enumerate}
\end{corollary}

\begin{proof}
\ref{cor:rank-sep:i}. 
The Euclidean distance matrix $M_m \in \mathcal{M}_m \cong \mathbb{R}^m \otimes \mathbb{R}^m$ which is defined as
$$\left(M_m\right)_{i,j} = (i-j)^2$$
fulfils (see \cite[Example 5.17]{Fa14} and \cite[Section 5]{De19d} for details) 
$$\rank_{\Lambda_1}(M_m) = 3,  \quad \psdrank_{\Lambda_1}(M_m) = 2, \textrm{ and } \quad \nnrank_{\Lambda_1}(M_m) \geq \log_2(m)$$
since all explicit examples are given as a real matrix factorization.
Defining $p_m \coloneqq p_{M_m}$ and using Proposition \ref{pro:OmegaG-corr-tensors} shows the statement.

\ref{cor:rank-sep:ii} is similar to \ref{cor:rank-sep:i}, this time using the slack matrix of an $m$-gon  (see \cite[Example 5.14]{Fa14}).
\end{proof}

These statements imply that there cannot exist functions $f,g \colon \mathbb{N} \to \mathbb{N}$ such that
\be \nn \seprank_{\Lambda_1}(p) \leq f\big(\sosrank_{\Lambda_1}(p)\big) \quad \textrm{ and } \quad \sosrank_{\Lambda_1}(p) \leq g\big(\rank_{\Lambda_1}(p)\big)\ee
holds for all $m \in \mathbb{N}$ and all polynomials $p \in \mathbb{R}[\mathbf{x}^{[0]}, \mathbf{x}^{[1]}]$ with $\mathbf{x}^{[i]} \coloneqq (x^{[i]}_1, \ldots, x^{[i]}_m)$.
This also holds true for polynomials of bounded degree, since   $\deg(p_m) = 4$ in the above construction.

But this also  immediately leads to the   question of whether there are separations between the ranks of polynomials with a bounded number of variables and no bound on the degree. 
In this setting  there does not exist a one-to-one correspondence between polynomials and  Gram matrices (as that of Example \ref{ex:sepPoly}). 
We believe that separations will again appear in the simplest setting   and leave this question as a conjecture.

\begin{conjecture}
\label{conj:sep}
There exist no functions $f,g,h: \mathbb{N} \to \mathbb{N}$ 
such that for all $p \in \mathbb{R}[x,y]$  (in particular, independently of the degree of $p$) 
\begin{enumerate}[label=(\roman*),ref=(\roman*),leftmargin=*]
	\item \label{conj:sep:i}
	$\seprank_{\Lambda_1}(p) \leq f(\rank_{\Lambda_1}(p))$
		\item \label{conj:sep:ii}
		$\seprank_{\Lambda_1}(p) \leq g(\sosrank_{\Lambda_1}(p))$
		\item \label{conj:sep:iii}
		$\sosrank_{\Lambda_1}(p) \leq h(\rank_{\Lambda_1}(p))$
\end{enumerate}  
where $p$ is of course separable in \ref{conj:sep:i} and \ref{conj:sep:ii}, and a sum of squares in \ref{conj:sep:iii}. The separable rank is again meant with respect to the local sos-cones.
\end{conjecture}

\section{Approximate polynomial decompositions: Disappearance of separations}
\label{sec:approx}

In this section we study $(\Omega,G)$-ranks of homogeneous polynomials in the approximate case. 
To this end, we will first show that approximations of polynomials  can be related to approximations of matrices (Lemma \ref{lem:gramapprox}), and will  leverage this result together with those of \cite{De20} to obtain  approximations for invariant separable polynomials (Theorem \ref{thm:approxSeparable}). 

To this end, we start by considering homogenized polynomials from $\mathcal{P}_d$. 
More specifically, restricting to $\deg_{\mathrm{loc}}(p) = d$, we study approximations in the space
\be \nn \mathcal{P}^{h}_{d} \coloneqq \mathbb{R}[\mathbf{x}^{[0]}]_d^{h} \otimes \cdots \otimes \mathbb{R}[\mathbf{x}^{[n]}]_d^{h} \ee
where $\mathbb{R}[\mathbf{x}]_d^{h}$ is the space of homogeneous polynomials of degree $d,$ and each $\mathbf{x}^{[i]} = (x^{[i]}_0, \ldots, x^{[i]}_m)$ is a vector of $m+1$ variables (note that we have introduced additional variables $x^{[i]}_0$ in contrast to $\mathcal{P}_d$).
Each homogeneous polynomial $q \in \mathcal{P}_d^{h}$ corresponds to some $p \in \mathcal{P}_d$ by setting the variables $x^{[0]}_0, \ldots, x^{[n]}_0$ to $1$. On the other hand, every $p \in \mathcal{P}_d$ can be multi-homogenized by substituting
\be \nn \left(x_{1}^{[i]}\right)^{\alpha_1} \cdots \left(x_{m}^{[i]}\right)^{\alpha_m} \mapsto \left(x_0^{[i]}\right)^{d - \vert\alpha\vert} \cdot \left(x_{1}^{[i]}\right)^{\alpha_1} \cdots \left(x_{m}^{[i]}\right)^{\alpha_m}\ee
for every local monomial in $\mathfrak{m}_d(\mathbf{x}^{[i]})$. 
For the rest of the section, we denote the basis of homogenized monomials by
\be \nn \mathfrak{m}^{h}_{n,d}(\mathbf{x}) \coloneqq \mathfrak{m}^{h}_d(\mathbf{x}^{[0]}) \otimes \cdots \otimes \mathfrak{m}^{h}_d(\mathbf{x}^{[n]})\ee
where each $\mathfrak{m}^{h}_d(\mathbf{x}^{[i]})$ is the vector of all monomials $\left(\textbf{x}^{[i]}\right)^\alpha$ with $\vert\alpha\vert = d$.
 We will  also consider  the  Gram map  
\begin{align*}
 \mathcal{G}:  \mathcal{M}_D^{\otimes n+1} & \to \mathcal{P}^{h}_{d} \\
 M &\mapsto (\mathfrak{m}_{n,d}^h)^t M \mathfrak{m}_{n,d}^h, 
\end{align*}
and the set $$\mathbb{S} \coloneqq \mathbb{S}^{m} \times \mathbb{S}^{m} \times \cdots \times \mathbb{S}^{m},$$ where the product runs over the set $[n]$, and where $\mathbb{S}^{m} \subseteq \mathbb{R}^{m+1}$ is the unit sphere with respect to the Euclidean norm. 

We will consider approximations of the polynomial $p$ with respect to the maximal value attained among all   $\mathbf{a} \coloneqq (\mathbf{a}^{[0]}, \ldots, \mathbf{a}^{[n]}) \in \mathbb{S}$. 
More specifically, we define the \emph{infinity norm} of $p$ as
$$
\Vert p \Vert_{\infty} \coloneqq \sup_{\mathbf{a} \in \mathbb{S}} \vert p(\mathbf{a}^{[0]}, \ldots, \mathbf{a}^{[n]})\vert. 
$$
It is easy to check, that $\Vert \cdot \Vert_{\infty}$ satisfies the triangle inequality and is absolutely scalable. 
In addition, $\Vert \cdot \Vert_{\infty}$ is positive definite since, if $\Vert p \Vert_{\infty} = 0$ and $\textbf{b}^{[0]}, \ldots, \textbf{b}^{[n]} \in \mathbb{R}^{m}\setminus\{0\}$, then there exist $\lambda_0, \ldots, \lambda_n \in \mathbb{R}$ and $(\textbf{a}^{[0]}, \ldots, \textbf{a}^{[n]}) \in \mathbb{S}$ such that
$$(\textbf{b}^{[0]}, \ldots, \textbf{b}^{[n]}) = (\lambda_0 \textbf{a}^{[0]}, \ldots, \lambda_n \textbf{a}^{[n]}).$$
Multi-homogeneity of $p$ implies $$p(\textbf{b}^{[0]}, \ldots, \textbf{b}^{[n]}) = \lambda_0^d \cdots \lambda_n^d \cdot p(\textbf{a}^{[0]}, \ldots, \textbf{a}^{[n]}) = 0,$$
which clearly implies $p = 0$.
Hence, $\Vert \cdot \Vert_{\infty}$ is a norm  on the space $\mathcal{P}^{h}_{d}$.

We start with the following preparatory lemma, which relates the infinity norm of $p$ with the Schatten $2$-norm of its Gram matrix $M$.

\begin{lemma}[Norm of polynomial and of Gram matrix]
\label{lem:gramapprox}
Let $\mathbf{a} \in \mathbb{S}$. Then
$\Vert \mathfrak{m}^{h}_{n,d}(\mathbf{a})\Vert_2 \leq 1.$
Moreover,
\be \nn \Vert \mathcal{G}(M) \Vert_\infty \leq \sigma_{\max}(M) \leq \Vert M \Vert_2\ee
where $\sigma_{\max}(M)$ denotes the maximal singular value of $M$. 
\end{lemma}

\begin{proof}
For the first statement, it suffices to show $\Vert \mathfrak{m}^{h}_d(\mathbf{a}^{[i]}) \Vert_2 \leq 1$ for each $i \in [n]$, as the statement then follows from the multiplicativity of the $2$-norm with respect to elementary tensors. 
We show  it  by induction over $d$. For $d=1$ we have $\mathfrak{m}_d(\mathbf{a}^{[i]}) = \mathbf{a}^{[i]}$, hence there is nothing to show. For $d \geq 1$ we have
\be \nn \Vert \mathfrak{m}^{h}_{d+1}(\mathbf{a}^{[i]}) \Vert_2^2 &=& \sum_{\vert\alpha\vert = d+1} \left(\mathbf{a}^{[i]}\right)^{2\alpha} \\ \nn &=& \sum_{j=0}^{m} \left(a^{[i]}_j\right)^2 \cdot \sum_{\substack{\vert\alpha\vert=d\\ \alpha_0, \ldots, \alpha_{j-1} = 0}}  \left(\mathbf{a}^{[i]}\right)^{2\alpha} \leq \sum_{j=0}^{m} \left(a^{[i]}_j\right)^2 \leq 1. \ee
where we have used the induction hypothesis  in the first inequality.

For the second statement, we have
\be \nn \Vert \mathcal G(M) \Vert_\infty &=& \sup_{\mathbf{a} \in \mathbb{S}} \vert\mathfrak{m}^{h}_{n,d}(\mathbf{a})^t M \mathfrak{m}^{h}_{n,d}(\mathbf{a})\vert \\
&\leq& \nn \sup_{\substack{y \in \mathbb{R}^{r} \\ \Vert y \Vert_2 \leq 1}} \vert y^t M y\vert = \sigma_{\max}(M) \leq \Vert M \Vert_2\ee
where  $r = D^{n+1}$ (where $D= \binom{m+d}{d}$) 
and the first inequality follows from the first statement.
\end{proof}

Recall that a separable matrix $M \in \mathcal{M}_d \otimes \cdots \otimes \mathcal{M}_d$ attains a separable $(\Omega,G)$-decomposition if it can be written as
$$ M = \sum_{\alpha \in \mathcal{I}^{\widetilde{\mathcal{F}}}} M_{\alpha_{\vert_0}}^{[0]} \otimes \cdots \otimes  M_{\alpha_{\vert_n}}^{[n]}$$
where $M^{[i]}_{\beta} \in \mathcal{M}_d^{+}$ is a real positive semidefinite matrix for every $i \in [n]$ and $\beta \in \mathcal{I}^{\widetilde{\mathcal{F}}_i},$ and
$$ M^{[gi]}_{{}^g\beta} = M^{[i]}_{\beta}$$
for every $i \in [n]$ and $g \in G$. For a more detailed study of this decomposition we refer to \cite{De19d}. To show the approximation result, we will exploit the following result from \cite[Proposition 24]{De20} about approximate $(\Omega,G)$-decompositions about normalized separable matrices.

\begin{proposition}[Approximate invariant decompositions \cite{De20}]
\label{pro:approxOmegaG}
Let $\Omega$ be a weighted simplicial complex  with a free group action from $G$, and fix $\varepsilon >0$. Let $M \in \mathcal{M}_{D}^{\otimes n+1}$ be $G$-invariant and separable with $\tr(M) \leq 1$. 
Then there exists a separable $N \in \mathcal{M}_{D}^{\otimes n+1}$  such that $\Vert M - N \Vert_2 < \varepsilon$ and
\be \nn \seprank_{(\Omega,G)}(N) \leq \left\lceil \frac{8\exp(4)}{\varepsilon^2}\right\rceil \cdot \vert G\vert. \ee
\end{proposition}

To guarantee that a given polynomial fulfils the normalization in Proposition \ref{pro:approxOmegaG} we introduce the following norm for $p \in \mathcal{P}^h_d$. Denote the set of (sub-)normalized separable matrices by
\be \nn \mathrm{SEP}_{n,D} \coloneqq \left\{M \in \mathcal{M}_{D}^{\otimes n+1}: M \textrm{ is separable and } \tr(M) \leq 1 \right\}\ee
and define
\be \nn \mu(p) \coloneqq \inf \left\{\lambda > 0: \exists M \in \mathrm{SEP}_{n, D} \textrm{ and $G$-invariant such that } p = \lambda \mathcal{G}(M) \right\}. 
\ee
Note that, by Remark \ref{rem:sep}, $\mu(p)$ is finite for all separable and $G$-invariant polynomials. Moreover, $\mu$ is homogeneous of degree $1$, i.e.\ for all $\gamma \geq 0$ we have $\mu(\gamma p) = \gamma \mu(p)$, and since $\mathrm{SEP}_{n,D}$ is convex we have $\mu(p_1 + p_2) \leq \mu(p_1) + \mu(p_2)$.

We can finally present the main result of this section. Recall that the separable decomposition and rank refer to the local sos cones.

\begin{theorem}[Approximate separable invariant decomposition]
\label{thm:approxSeparable}
Let $\Omega$ be a weighted simplicial complex  with a free group action from $G$,  and fix  $\varepsilon >0$. Further, let $p \in \mathcal{P}_d^{h}$ be separable and $G$-invariant. Then there exists $q \in \mathcal{P}_d^{h}$ such that $\Vert p - q \Vert_{\infty} < \varepsilon$ and
\be \nn \seprank_{(\Omega,G)}(q) \leq \left\lceil \frac{8\exp(4) \cdot \mu(p)^2}{\varepsilon^2}\right\rceil \cdot \vert G\vert.\ee
\end{theorem}

\begin{proof}
Since $p$ is $G$-invariant and separable, there exists a separable and $G$-invariant matrix $M$ such that $p = \mathcal{G}(M)$ by Remark \ref{rem:sep}. Choose $M$ so that $\tr(M)$ is minimal among all representations. Then
\be \nn 
\frac{1}{\tr(M)} M \in \mathrm{SEP}_{n,D} \: \textrm{ and } \: \mu(p) = \tr(M).
\ee
By Proposition \ref{pro:approxOmegaG} there exists $N \in \mathcal{M}_{D}^{\otimes n+1}$ separable such that $\Vert M - N\Vert_2 < \varepsilon$ and
$$
\seprank_{(\Omega,G)}(N) \leq \left\lceil \frac{8\exp(4) \mu(p)^2}{\varepsilon^2}\right\rceil \cdot \vert G\vert.
$$
Now define $q = \mu(p) \cdot \mathcal{G}(N)$.
By Lemma \ref{lem:gramapprox} we have that $\Vert p - q\Vert_{\infty} < \varepsilon$, and since the Gram map applied to an $(\Omega,G)$-decomposition of matrices leads to an $(\Omega,G)$-decomposition of polynomials, we obtain that $\seprank_{(\Omega,G)}(q) \leq \seprank_{(\Omega,G)}(N),$ which proves the statement.
\end{proof}

Theorem \ref{thm:approxSeparable} provides an upper bound of $\seprank_{(\Omega,G)}$ which is (up to $\mu(p)$) dimension-independent. This implies that the separations between $\rank_{(\Omega,G)}$, $\sosrank_{(\Omega,G)}$ and $\seprank_{(\Omega,G)}$ \emph{disappear in the approximate case} if the value of  $\mu(p)$ is bounded. 
In general, however, $\mu(p)$ scales  with $\deg_{\mathrm{loc}}(p)$ and the number of variables $m$.

Similar approximation procedures can be applied to sos polynomials together with sos $(\Omega,G)$-decompositions, or arbitrary polynomials together with unconstrained $(\Omega,G)$-decompositions. 
This can be accomplished by exploiting  approximation results of $(\Omega,G)$-decompositions  for positive semidefinite matrices and Hermitian matrices \cite{De20}. 
Together with the norm correspondence from Lemma \ref{lem:gramapprox}, this would lead to approximations for all types of polynomial $(\Omega,G)$-decompositions, that we decided not to work out in full generality here.

\section{An undecidable problem regarding unconstrained decompositions}
\label{sec:undecidable}

In Section \ref{sec:inequalitites} we have seen that the invariant sos decomposition and the invariant separable decompositions, which are inherently positive, are generally much more costly than the decomposition without any positivity constraints on the local elements. 
Here we will show that the invariant separable decomposition has in fact no local and computable certificate of positivity. 
We will reach this conclusion  by proving that Problem \ref{prob:und} is undecidable. 

Given a collection of $D^2$ polynomials in $\mathbb{Z}[\mathbf{x}]$, denoted $\left(p_{\alpha,\beta}\right)_{\alpha,\beta=1}^{D}$, define 
\be\label{eq:pn}
p_n \coloneqq \hspace*{-0.4cm} \sum_{\alpha_0, \ldots, \alpha_n = 1}^{D} p_{\alpha_0, \alpha_1}(\mathbf{x}^{[0]}) \cdot p_{\alpha_1, \alpha_2}(\mathbf{x}^{[1]}) \cdots p_{\alpha_n, \alpha_0}(\mathbf{x}^{[n]})\in\mathbb R[\mathbf x^{[0]},\ldots, \mathbf x^{[n]}].
\ee 
Note that the summation indices are arranged in a circle $\Theta_n$, and that the local polynomials are independent of site, so that $p_n$ is invariant under the cyclic group $C_n$. The previous expression is thus a $(\Theta_n,C_n)$-decomposition of $p_n$. 

\begin{problem}[Decision problem about positivity of polynomials]
\label{prob:und}
Given positive integers $m$ and $D$ and a collection of polynomials  $\left(p_{\alpha,\beta}\right)_{\alpha,\beta=1}^{D} \in \mathbb{Z}[\mathbf{x}]$ (where  $\mathbf{x}$ denotes a vector of $m$ variables $(x_1, \ldots, x_m)$), 
\begin{enumerate}[label=(\alph*),ref=(\alph*),leftmargin=*]
	\item \label{prob:und:a} Is $p_n$  a  sum of squares for all $n \in \mathbb{N}$?
	\item \label{prob:und:b} Is $p_n$    nonnegative for all $n \in \mathbb{N}$?
\end{enumerate}
\end{problem}

\begin{theorem}[Undecidability of Problem \ref{prob:und}]
\label{thm:und}
Problem \ref{prob:und} \ref{prob:und:a} and Problem \ref{prob:und} \ref{prob:und:b} are undecidable.
This is true even if $m,D \geq 7$  and if  the polynomials are of the form
$$p_{\alpha,\beta}(\mathbf{x}) = \sum_{j=1}^{m} p_{\alpha, \beta, j} \cdot x_{j}^2 
$$
with $p_{\alpha, \beta, j} \in \mathbb{Z}$ for all $\alpha, \beta \in \{1, \ldots, D\}$. 
\end{theorem}

So there does not exist an algorithm that can decide in finite time whether $p_n$ is sos or nonnegative for all $n$, given the local polynomials as input. (For an introduction to undecidability we refer for example to \cite{Ar09}.)
We will prove Theorem \ref{thm:und} by a  reduction from the following undecidable problem:

\begin{theorem}[Undecidability of positivity for all system sizes \cite{De15}]\label{thm:undold}
Let $T_{\alpha,\beta} \in \mathbb{Z}^m$ for $\alpha, \beta \in \{1, \ldots, D\}$ be a collection of vectors.
For $n\geq 0$ define 
$$ T_n \coloneqq \sum_{\alpha_0, \ldots, \alpha_n = 1}^{D} T_{\alpha_0, \alpha_1} \otimes T_{\alpha_1, \alpha_2} \otimes \cdots \otimes T_{\alpha_n, \alpha_0}.$$
For $m, D \geq 7,$ the following problem is undecidable:
$$ \textrm{Is }  T_n \textrm{ nonnegative for all } n \in \mathbb{N}?$$
\end{theorem}

\begin{proof}[Proof of Theorem \ref{thm:und}]
Let $T_{\alpha, \beta} \in \mathbb{Z}^m$ be a collection of vectors for $\alpha, \beta \in \{1, \ldots, D\}$.  We apply the construction from Section \ref{ssec:sep} to obtain the collection of polynomials $p_{\alpha,\beta}= \sum_{j=1}^{m} \left(T_{\alpha, \beta}\right)_{j} x_j^2$ and generate the polynomials $p_n \in \mathbb{Z}[\mathbf{x}^{[0]}, \ldots, \mathbf{x}^{[n]}]$. It is obvious that $p_{T_n}=p_n$   for all $n$,
and from Lemma \ref{lem:correspondencePositivity} we thus know that $T_n$ is nonnegative if and only if $p_n$ is a sum of squares/nonnegative.  
So decidability of Problem \ref{prob:und} \ref{prob:und:a} or \ref{prob:und:b}  contradicts Theorem \ref{thm:undold}.
\end{proof}

We remark that Problem \ref{prob:und} remains undecidable if  the input polynomials are in $\mathbb{Q}[\mathbf{x}]$, since multiplying all polynomials by a positive constant does not change the positivity/sos property. 

It can also be shown that a \emph{bounded} version of the questions of Problem \ref{prob:und}---i.e.\ where $n$ is fixed---result in an \textsf{NP}-hard problem \cite{Kl14}.  

\section{Conclusions and Outlook}
\label{sec:concl}

In summary, we have defined and studied several decompositions of multivariate polynomials into local polynomials, each containing only a subset of variables. 
The variables are divided into blocks, and each local polynomial uses only one block. 
We describe a decomposition with a weighted simplicial complexes $\Omega$, 
whose vertices describe the individual blocks, and facets the summation indices.  
For polynomials invariant under the permutation of blocks of variables, 
we have defined and studied an invariant decomposition. 
We have also defined an invariant decomposition with local positivity conditions, specifically, with the separable and sum of squares condition. 
Our approach is inspired by the tensor network approach from quantum information theory; 
in particular, the framework of this work was previously applied to tensor decompositions \cite{De19d} and   studied in the approximate case in Ref.\ \cite{De20}.

Specifically, we have defined invariant polynomial decompositions (Definition \ref{def:omegaG-dec}) 
and shown that every $G$-invariant polynomial admits an $(\Omega,G)$-decomposition if $G$ acts freely on $\Omega$ (Theorem \ref{thm:omegaG-dec}), 
and that every group action can be made free by increasing the number of summation indices (Proposition \ref{pro:extend}). 
Moreover, if $G$ is a blending group action, every $G$-invariant polynomial can be written as a difference of two $(\Omega,G)$-decompositions (Theorem \ref{thm:omegaG-diff}).
We have also defined the separable $(\Omega,G)$-decomposition  (Definition \ref{def:sepdec}), 
and sum of squares $(\Omega,G)$-decomposition (Definition \ref{def:sos-dec}), 
and have shown that they exist if $G$ acts freely on $\Omega$ (Theorem \ref{thm:sep} and Corollary \ref{cor:sos}, respectively). 

In addition, we have shown that  the $(\Omega,G)$-rank of a polynomial can be upper bounded in terms of 
its separable and sos rank, 
and that the sos rank can often be upper bounded by its  separable rank (Proposition \ref{pro:inequalitites}). 
In the reverse direction such inequalities cannot exist, since there exists a sequence of polynomials with constant $(\Omega,G)$-rank and a diverging sos or separable rank (Corollary \ref{cor:rank-sep}). 
Yet, these separations are not robust with respect to approximations, 
due to the upper bound of the approximate separable invariant decomposition 
provided in Theorem \ref{thm:approxSeparable}.
Finally, for decompositions on the circle with translational invariance, 
we have shown it is undecidable whether the global polynomial is sos or nonnegative for all system sizes (Theorem \ref{thm:und}).

This work has left two ``immediate'' open questions: 
Whether the rank separations also hold with respect to a bounded number of variables but unbounded degree (Conjecture \ref{conj:sep}), and  whether there exist non-factorizable  $(\Omega,G)$ structures  (Question \ref{qu:fact}).
A more general open question concerns the full characterization of the existence of invariant polynomial decompositions, 
as  freeness of the group action only provides a sufficient condition. 
Our investigations indicate that it may also be  necessary, but we were not able to prove it. 

A very interesting question is: What is the \emph{border rank} of an $(\Omega,G)$-decomposition? 
The border rank provides a complementary notion of approximation than the one considered here, and shows surprising features for tensors (instead of matrices). The $(\Omega,G)$-framework is an invitation to generalise this study to tensor decompositions on $\Omega$, possibly with invariance. 

Our existence theorems work for a given system size $n$. 
What can be said about \emph{all} system sizes? 
Namely, if a family of objects (such as tensors or polynomials) is invariant for each system size, does it admit a \emph{uniform} invariant decomposition? 
The undecidability result of Theorem \ref{thm:und} suggests that this question is very different from the one addressed in this paper, but certainly very interesting. 

\section*{Acknowledgments}
This research was funded in part by the Austrian Science Fund (FWF) [doi:\href{https://www.doi.org/10.55776/P33122}{10.55776/P33122}]. For open access purposes, the authors have applied a CC BY public copyright license to any author accepted manuscript version arising from this submission. AK further acknowledges funding of the Austrian Academy of Sciences (\"OAW) through the DOC scholarship 26547.



\end{document}